\documentclass[review]{elsarticle}

\usepackage{graphicx}
\usepackage{amssymb}
\usepackage[all]{xy}
\newdir{ >}{{}*!/-6pt/@{>}}
\usepackage{paralist}
\usepackage{eurosym}
\usepackage[utf8]{inputenc}
\usepackage{enumitem}
\usepackage{caption} 
\usepackage{subcaption} 
\usepackage{stmaryrd}
\newtheorem{definition}{Definition}%
\newtheorem{example}{Example}%
\newtheorem{theorem}{Theorem}%
\newtheorem{fact}{Fact}
\newtheorem{remark}{Remark}
\newtheorem{corollary}{Corollary}%
\newtheorem{construction}{Construction}%
\newtheorem{proposition}{Proposition}%
\newenvironment{proof}{\noindent {\bf Proof:}}
{\hfill$\square$\vskip .2cm}%

\journal{Blockchain: Research and Application}

\def\C{{\cal C}}

\def\E{{\cal E}}

\def\X{{\cal X}}
\def\Y{{\cal Y}}

\def\implies{\Rightarrow}

\DeclareRobustCommand{\bitcoin}{{%
		\normalfont\sffamily
		\raisebox{-.05ex}{\makebox[.1\width][l]{-\kern-.2em-}}B%
}}

\newcounter{ncomm}%

\newcounter{pcomm}%
\newcommand{\cnt}[1]{\Gamma\!\!}
\newcommand{\cntxt}[1]{\begin{tiny}\Gamma\end{tiny}}

\begin{document}

\begin{frontmatter}

\title{A formal model for ledger management systems based on contracts and temporal logic}

\author{Paolo Bottoni, Anna Labella}
\address{Department of Computer Science - Sapienza University of Rome - Italy}
\author{Remo Pareschi}
\address{Stake Lab - University of Molise - Italy}

\begin{abstract}
A key component of blockchain technology is the ledger, viz., a database that, unlike standard databases, keeps in memory the complete history of past transactions as in a notarial archive for the benefit of any future test. 
In second-generation blockchains such as Ethereum the ledger is coupled with smart contracts, which enable the automation of transactions associated with agreements between the parties of a financial or commercial nature. 
The coupling of smart contracts and ledgers provides the technological background for very innovative application areas, such as Decentralized Autonomous Organizations (DAOs), Initial Coin Offerings (ICOs) and Decentralized Finance (DeFi), which propelled blockchains beyond cryptocurrencies that were the only focus of first generation blockchains such as the Bitcoin. 
However, the currently used implementation of smart contracts as arbitrary programming constructs has made them susceptible to dangerous bugs that can be exploited maliciously and has moved their semantics away from that of legal contracts.
We propose here to recompose the split and recover the reliability of databases by formalizing a notion of contract modelled as a finite-state automaton with well-defined computational characteristics derived from an encoding in terms of allocations of resources to actors, as an alternative to the approach based on programming.
To complete the work, we use temporal logic as the basis for an abstract query language that is effectively suited to the historical nature of the information kept in the ledger.
\end{abstract}


\end{frontmatter}


\section{Introduction}\label{sec:intro}
Blockchains and distributed ledgers are nowadays among the information technologies having the greatest impact. 
Initial success has come to public blockchains, starting with the mother of all of them, that of Bitcoin, which has cleared the large-scale practicability of cryptocurrencies, i.e., currencies freed from central issuing authorities and governed by the community of their users. 
Public blockchains of the next generation, such as Ethereum, have combined the community-based management of cryptocurrencies with the experimentation and practice of new forms of organization and finance, such as decentralized autonomous organizations (DAOs), initial coin offerings (ICOs) and decentralized finance (DeFi). 
Finally, private distributed ledgers, of which the best known and most practiced are the platforms developed within the Hyperledger open-source project, have released the potential for  organizational innovation of these technologies from the support of cryptocurrencies, thus making it possible to involve brick-and-mortar companies in innovative business ecosystems.

At the base of all this there is a fundamental technological component, referred to here as the ``digital ledger", which is nothing more than a database management system (DBMS) that preserves the history of all its records  by appending new versions of records to previous ones and linking them together through identifiers such as hash pointers. 
Several DBMSs that predate the coming of age of blockchains work that way, like the popular Hadoop file system (HDFS) largely adopted for big data management. 
Distributed ledgers are a special case of digital ledger where the validation of transactions takes place through a consensus mechanism over a network of peer nodes. 
Public blockchains, in turn, specialise distributed ledgers by grouping transactions on records into blocks built by nodes rewarded for the construction with the native currency of the blockchain. 
Private distributed ledgers are yet another specialization, as are public distributed ledgers based on data arrangements alternative to blocks, such as the IOTA platform, aimed at supporting Internet-of-Things applications, that uses instead directed acyclic graphs.

However, there is more, in the latest ledgers that have entered the arena in the wake of the blockchain boom, than just decentralized validation of transactions through distributed consensus protocols. 
There are, in fact, smart contracts: conceptually introduced by the multifaceted scholar Nick Szabo 
(jurist, cryptographer, computer scientist) at the end of the 1990s for computer-based automated execution of legal contracts of a commercial nature~\cite{Sza97}, the idea stayed dormant until it was dusted off from oblivion by Vitalik Buterin, the founder and creator of Ethereum, in the mid-10s of this century, who transferred it to the eponymous blockchain by equipping it for implementation with Solidity, a contract-oriented  programming language~\cite{But13}. 
Ethereum's move was soon repeated on 
%
%
other platforms, both public and private, and there is now ample opportunity 
%
for 
smart contracts in blockchain and distributed ledger projects. 

The reason why smart contracts lend themselves particularly well to implementation on a ``distributed ledger" is actually more on the ``ledger" than on the ``distributed" part of the term. 
In fact, the permanence of data on the ledger makes the execution of smart contracts traceable, which of course is an advantage for auditing purposes whenever need may arise, as in a due diligence or a litigation. 
On the other hand, also distributed consensus can support traceability, albeit less essentially, by multi-checking the execution of the agreements made through many validators, which strengthens the claim to veracity of the transcribed data in the absence, by choice or necessity, of a reliable central authority. 
In any case, what really matters is operating  in a context predisposed to auditing, such as a digital ledger by its nature is.
%

However, it is precisely the choices made for the implementation of smart contracts that have distanced the platforms that have adopted them both from the database characteristics of first generation distributed ledgers like the Bitcoin and from smart contracts as originally conceived. In fact, although referred to as a contract-oriented language, Ethereum's  Solidity is actually a scripting language that can write on the blockchain. Similarly, in other environments, such as Hyperledger, 
%
%
smart contracts are programs written in conventional languages that can access the ledger. 
These programming tools have certainly been useful in propelling the application of distributed ledgers beyond cryptocurrencies, but have also created reliability problems by making possible constructs of dubious semantics and arbitrary complexity; witness among all the notorious DAO exploit of 2016 in which a still unknown hacker exploited a bug in a Solidity smart contract to steal \$ 50 million from the Ethereum common pot. 
Moreover, they have weakened the relationship between smart contracts and legal contracts as originally conceived by Szabo and taken up in intent, but not in practice, by Buterin; indeed, there is very little in common between the text of a commercial contract and a program in Solidity.

The purpose of this article is to define a general ledger model with contracts seen as declarative constructs on a database, rather than as arbitrary procedural programs, with well-defined operational characteristics in the tradition of advanced transaction approaches such as the renowned ACID model. 
To this end, we propose a bare-bones view of contracts based on the notion of \emph{allocation} of resources to actors, so that a contract defines admissible evolutions of states of affairs defined by sets of such allocations. 

At the same time, our model aims to make legal contracts in textual form effectively transferable and translatable into executable contracts, thus recovering the original idea of smart contract as introduced by Szabo. 
Hence, we propose two forms of contract-related automata: the first identifies the possible states of the contract with respect to the completion of some sets of obligations, where the discharge of an obligation is modeled as the transfering of some resource from some actor to another one; 
%
%
the second provides a refined view of the first, taking into consideration the possible orderings in which obligations
can be discharged, thus turning a ledger into a faithful record of the sequence of transfers actually occurred during the execution of the contract.

On our way towards a model endowed with full-fledged DBMS characteristics, we will also deal with querying, by providing constructs 
%
%
to verify progress in a contract execution, a functionality that is badly needed in practice, as well as to query the ledger by scrolling it back and forth in time, with advanced query modes ranging beyond the still rather limited solutions currently available. 
The formal cornerstones of our model are the theory of automata and temporal logic, formalisms that have been both rigorously systemized and theorized. 
As for distribution and centralization, we will be agnostic about the issue, since our model can be adapted to either one of the two options. 
For the organization of the data, we will maintain the minimal requirement that they are organized as in a ledger. 
This does not preclude more specialized organizations, like blocks, the model being easy to extend and adapt  to ledgers structured as blockchains. 

The rest of this article is structured as follows: Section~\ref{sec:background} introduces some background notations and definitions.
Section~\ref{sec:contracts} provides the fundamental notions of \emph{resource}, \emph{actor} and \emph{transfer}, so that a contract can be seen as defining a set of constraints on sequences of transfers of resources between actors. 
Section~\ref{sec:automaton} shows how a legal contract understood as a set of interdependent obligations can be modeled 
%
%
in computational terms by defining admissible paths on some contract-related automaton, where transitions are associated with the discharge, through transfers of resources, of obligations. 
Section~\ref{sec:encoding} then shows how contract automata can be integrated with a ledger so that their actions are transcribed as records into the ledger, while Section~\ref{sec:categoricalModel} builds, from the algebraic structure of ledgers and contracts, modal/temporal logics used in Section~\ref{sec:language} to define an abstract query language that can be applied to extract information about both static (the records in the ledger) and dynamic 
(the states of contracts whose execution is in progress) aspects of a contract execution. 
Finally, Section~\ref{sec:related} discusses related work and Section~\ref{sec:concls} concludes the article.

\section{Background}\label{sec:background}
We recall some basic notions and notations, useful in the rest of the paper.

For $n\in\mathbb{N}$, the set of all integers from $1$ to $n$ is denoted by $[n]$.
An \emph{alphabet} is a finite set $A=\{a_1,\dots,a_k\}$, where each element $a_i$ is called a \emph{member} of $A$.
For $k\in\mathbb{N}$, a sequence $\omega=\langle{x}_1\cdots{x}_k\rangle$ of elements from $A$ is called a \emph{word} on $A$ of \emph{length} $k$. 
We use the notation $\omega[i]$ to indicate the element $x_i$ for $i\in[k]$.
We denote the length of a sequence $\omega$ by $|\omega|$ and the unique word of length $0$ by $\epsilon$. 
Then, the set of all words on $A$ 
%
%
(for each possible length)
%
%
is denoted by $A^*$.

Given a partially ordered set $(X,\le)$, if $(x,y)\in\le$, then $x$ is called a \emph{prefix} of $y$ and $y$ is called a \emph{prolongation} of $x$. 
Vice versa, given a prefix relation on a set $X$, its transitive closure defines a strict partial order.
%
%

It is easy to see that the above definition translates to the standard notion of prefix for the case of words. 
%
%
Indeed, given a word $\omega=\langle{x}_1\cdots{x}_k\rangle$, for each $l\le{k}$, 
the word $\omega_l=\langle{x}_1\cdots{x}_l\rangle$  is the \emph{prefix} of $\omega$ of length $l$, denoted by $pref(\omega,l)$. 
%
%
For any word $\omega$, $pref(\omega,0)=\lambda$. 
We denote by $PREF(\omega)$ the set $\{pref(\omega,l)\mid{l}\in\{0,\dots,|\omega|\}\}$.
$A^*$ results thus partially ordered according to the prefix relation, 
i.e. $\omega\le\omega^\prime$ if and only if $\omega\in{PREF}(\omega^\prime)$.

A partially ordered set $(L,\le)$ forms a \emph{meet-semilattice}, $\mathbf{L}=(L,\le,\wedge,\bullet)$, 
if it can be equipped with a \emph{meet} operation, $\wedge$ (i.e., greatest lower bound), and a minimum element, $\bullet$.
Then, a \emph{tree} is a set of elements, called \emph{paths}, in a meet-semilattice with the meet operation as a gluing function between them.
%
%

\begin{definition}[$\mathbf{L}$-tree]\label{def:LTrees}
Let $\mathbf{L}=(L,\le,\wedge,\bullet)$ be a  meet-semilattice. Then: 
\begin{itemize}[topsep=2pt,partopsep=1pt,itemsep=2pt,parsep=2pt]
\item A \emph{deterministic $\mathbf{L}$-tree} (\emph{tree} for short) is a pair $\X=(X,\wedge_X)$, where $X\subseteq{L}$ and $\wedge_X:X\times{X}\rightarrow{L}$ is the restriction to $X$ of the \emph{meet} $\wedge$.
\item An $\mathbf{L}$-tree is called \emph{prefix-closed} if it contains all the prefixes of its paths.
\item Given two $\mathbf{L}$-trees $\X=(X,\wedge_X)$ and $\Y=(Y,\wedge_Y)$, we say that $\X$ is a \emph{subtree} of $\Y$, 
noted $\X\subseteq\Y$, if $X\subseteq{Y}$ ($\subseteq$ denotes set-theoretical inclusion).
\end{itemize}
\end{definition}

Given a meet-semilattice $\mathbf{L}=(L,\le,\wedge,\bullet)$, the tree ${\cal{L}}=({L},\wedge)$ is canonically associated with $\mathbf{L}$ and 
formed by taking as set of paths the whole of $L$ and defining the gluing of data between paths as given by the $\wedge$ operation.

\begin{proposition}\label{prop:Boole}
Let $\mathbf{L}=(L,\le,\wedge,\bullet)$ be a meet-semilattice and ${\cal{L}}=({L},\wedge)$ its canonically associated tree. 
Then:
\begin{enumerate}[topsep=2pt,partopsep=1pt,itemsep=2pt,parsep=2pt] 
\item  a prefix-closed $\mathbf{L}$-tree is a meet-subsemilattice of $\mathbf{L}$;
\item the set of subtrees of ${\cal{L}}$, $Subtree({\cal{L}})$, with set-theoretical inclusion as partial order, is a boolean algebra (see~\cite{BGKL18}); and
\item given $\X\in{Subtree}({\cal{L}})$, the set of subtrees contained in $\X$, $Subtree(\X)$, is a boolean algebra in turn (see~\cite{BGKL18}).
\end{enumerate}
\end{proposition}

We define an important property of monotonic functions between posets, by specialising to this case the general notion of adjunction between functors.

\begin{definition}[Adjoints]\label{def:adjoints}
	Given two posets, $P=(X,\leq)$ and $P^\prime=(X^\prime,\leq^\prime)$, and two monotonic functions, 
	$f:P\rightarrow{P}^\prime$ and $g:P^\prime\rightarrow{P}$, 
	we say that $f$ is \emph{right adjoint to} $g$ ($g$ is \emph{left adjoint to} $f$) whenever it happens that: 
	for all $x\in{X}$ and  $x^\prime\in{X}^\prime$,  $x^\prime\leq{f}(x)$ if and only if $g(x^\prime)\leq{x}$.
\end{definition}

In other words: $f(x)$ is the least upper bound of the set $\{x^\prime\mid{g}(x^\prime)\le{x}\}$.
Dually, $g(x^\prime)$ is the greatest lower bound of the set $\{x\mid{x}^\prime\le{f}{x}\}$. 
In the case of posets, an adjunction is also called a \emph{Galois connection}.

\begin{fact}\label{fac:adjoints}
The following hold:
\begin{enumerate}[topsep=2pt,partopsep=1pt,itemsep=2pt,parsep=2pt]
\item An adjoint to a given function, if it does exist, is unique, hence it is characterised by this property.
\item Composition of two adjoints on the same side is still an adjoint on the same side. We say that the first one is preserved by the second one.
\item Boolean operators are adjoints on one side ($\wedge, \implies$ on the right, $\vee, \perp$ on the left), so that they are preserved by operators adjoint on the same side. 
In this sense a monotonic function between Boolean algebras which is a one-side adjunction, 
will provide what we call a \emph{smooth translation} from one algebra to the other one, because it will preserve part of its structure; 
\emph{a fortiori} if it is a two-side adjunction.
\end{enumerate}
\end{fact}



\section{Resources, actors, transfers}\label{sec:contracts} 
In this section, we introduce a bare-bones view of contracts, seen as constructs imposing constraints on admissible sequences of transfers of resources (from some given set), among actors (from some given set), entitled to some form of ownership on these resources. 
The proposed model was first formulated in~\cite{BGKL19} and fully developed within the theory of reaction systems in~\cite{AL21}.

The definition of these sets can occur either extensionally or intensionally. 
An example of a first case is provided by a loan contract, where both the lender and borrower parties are identified by names, the lent resource is identified by means of some title of property, the transfer of the right to its use by the possession of a copy of the text of the contract itself, and the number of instalments to be payed is defined and timestamped.
An example of the second case is provided by bearer bonds, where the assets are identified, but the only identified actor is the emitter of the bond, or, vice versa, by a Memorandum of Understanding, whereby two actors commit to share future products, which are only partially defined, e.g., by mentioning their types, at the time of the contract.
 
The constraints set by the contract, on the other hand, range from general, overarching conditions, such as: ``an actor $k$ cannot transfer a resource $r$ in a situation where $k$ is not entitled to the use of $r$'', to specific ones, as in the mentioned case of regular payments of a loan.

Another category of constraints may impose some kind of transactionality, such that a certain set of transfers have to occur "simultaneously" among some actors.
In the loan example, with each payment received from the borrower, the lender must produce a receipt for it and give it to the borrower.
Moreover, the execution of some transfer can be conditional on the occurrence of some event, 
for example the expiration of a deadline for payment of instalments, or a car accident for starting a damage compensation procedure.

Regardless of these differences, we assume that the general form of a transfer can be expressed as 
``actor $k1$ yields resource $r$ to actor $k2$'', 
%
%
assuming that $r$, $k1$, and $k2$ are all unambiguously identified at the time the transfer occurs. 

Moreover, as we are interested here in the encoding of transfers on a digital ledger, we assume that each of $r$, $k1$, and $k2$ is suitably represented by some URI, corresponding, respectively, to a digital token representing the asset $r$, or to (possibly encrypted) accounts in a digital store associated with the ledger.

\begin{example}\label{ex:contract}
Let us consider the case of a contract binding \emph{AL} to sell a house, say \emph{house1} 
(i.e., to transfer some property document \emph{house1PropDoc}) to \emph{PB}, 
and \emph{PB} to pay a certain amount, say \euro 500K, to \emph{AL}. 
The actual payment (in the ``real'' world) is mediated through some identifiable resource, e.g., a set of banknotes, a cashier's cheque from \emph{PB}'s account in the name of \emph{AL}, a certain amount of bitcoins in \emph{PB}'s wallet, etc., that we represent here abstractly as the unique token \emph{PayKE500Doc1}, testifying that the payment has occurred 
(this is equivalent to, say, having a copy of the cheque taken in the real world and registered by the notar).
Then, the two transfers, of \emph{house1PropDoc} from \emph{AL} to \emph{PB} and of \emph{PayKE500Doc1} from \emph{PB} to \emph{AL}, constitute a \emph{transaction}, as they must occur ``at the same time'', meaning that any observation of the execution of the contract which reports the first exchange must also be able to report on the second exchange, and vice versa. 
We say that these two transfers are \emph{co-occurent}.
On the other hand no transfer of  
\emph{PayKE500Doc1} from \emph{PB} to some $k_1$ other than ${AL}$ (since the token is produced for this specific transaction) or of \emph{house1PropDoc} from \emph{AL} to some $k_2$ other than  \emph{PB} (since the house cannot be sold two times) can be co-occurrent with the previous two.
\end{example}

To sum up, we view any, virtual or tangible, asset or service mentioned in a contract, and whose creation, consumption, or \emph{transfer}, is relevant to the contract, as a \emph{resource}. 
More precisely, we consider that at each moment some \emph{token} representing some form of (possibly shared) ownership of a resource is allocated to some actor also mentioned in the contract. 
Instances of such tokens are the payment instrument and the property document mentioned in Example~\ref{ex:contract}.

In this paper, as we abstract away from the actual form and nature of the resources, being only interested in how the corresponding tokens are distributed among actors at any given moment (we call such a distribution a \emph{state of affairs}), we use the terms (resource) ``token'' or ``resource'', indifferently.
Hence, with each contract, we associate a  nonempty set, $\mathbf{R}$, of resource tokens, constituting the overall universe of discourse. 

The set ${\mathbf{R}}$ can be constructed so as to accommodate quotas of some bulk resource, in analogy with shares of a company. 
So, for example, the right to use 100 liters of water delivered by \emph{WaterInc.} can be represented as $[water,WaterInc,100,R]$, where $R$ is a unique identifier generated on the fly.

On the other hand, we see an \emph{actor} as representing an individual entity bound by a contract and participating in transfers of resource tokens, whether at the yielding or at the receiving end. 
We assume that for any given contract the actors who may participate in it cannot be also regarded as resources, 
so that they are modeled as a nonempty set 
${\mathbf{K}}$, with ${\mathbf{K}}\cap{\mathbf{R}}=\emptyset$. 

We model the information that a token $r$ is held by some actor $k$ as the \emph{allocation} of $r$ to $k$, noted $[r,k]$.
A \emph{state of affairs on} $(\mathbf{R},\mathbf{K})$ is then defined as a set of allocations such that each token in $\mathbf{R}$ is allocated to one and only one actor in $\mathbf{K}$.
Hence, we model a state of affairs $\varsigma$ as the graph of a map $s_\varsigma:\mathbf{R}\rightarrow\mathbf{K}$, 
The set of all states of affairs on $(\mathbf{R},\mathbf{K})$ is denoted by ${\cal{S}}(\mathbf{R},\mathbf{K})$. 
Note that, according to the discussion above, although each token is unique, we can model joint ownership of a resource by generating a different token for each quota of that resource to which an actor is entitled. 

A state of affairs evolves through exchanges of tokens among actors, possibly as part of complex transactions.
Thus, the yielding of a token $r$ by an actor $k_1$, originally holding it, to an actor $k_2\neq{k}_1$, which becomes the holder of $r$, constitutes a \emph{transfer}, noted $\theta=(r,k_1,k_2)$. 
Then, 
%
%
$r$ is the \emph{transferred resource}, denoted by \emph{res}$(\theta)$.
%
%
The set of all transfers on $(\mathbf{R},\mathbf{K})$ is denoted by ${TRA}(\mathbf{R},\mathbf{K})$. 

Although the ``real'' transfer of resources may occur physically in some way not controlled by a program, a distributed ledger trusted to maintain information on the contract must be updated about such transfers as they occur, i.e., it has to maintain the record of the transfers undergone by the corresponding tokens.
At any time instant, the current allocation of a resource can thus be reconstructed by following the sequence of transfers for its associated token(s).

To this end, we say that a transfer $\theta=(r,k_1,k_2)\in{TRA}(\mathbf{R},\mathbf{K})$ is \emph{applicable in} a state of affairs $\varsigma\in{\cal{S}}(\mathbf{R},\mathbf{K})$ if and only if $[r,k_1]\in\varsigma$. 
Then, $APL(\varsigma)$ is the set of transfers applicable in $\varsigma$.
Given $\mathbf{R}$ and $\mathbf{K}$, for $\theta=(r,k_1,k_2)\in{TRA}(\mathbf{R},\mathbf{K})$ and $\varsigma\in{\cal{S}}(\mathbf{R},\mathbf{K})$, such that
$\theta\in{APL}(\varsigma)$, the \emph{application} of $\theta$ \emph{to} $\varsigma$, noted $apl_\theta(\varsigma)$,
produces
$\varsigma^\prime=(\varsigma\setminus\{[r,k_1]\})\cup\{[r,k_2]\}\in{\cal{S}}(\mathbf{R},\mathbf{K})$. 

As stated before, contracts may require a given set of transfers to be co-occurrent. 
Their application must then occur in a \emph{transactional} way, 
i.e., all transfers in the set are applied to a given state of affairs $\varsigma$ if they are all applicable in $\varsigma$, and no transfer in the set is applied if any of them is not applicable in $\varsigma$.
Due to the overall constraint prohibiting multiple transfers of the same token from the same actor, and the constraint that each token can be held by only one actor at a time, the transfers in the set must refer to different resources.

Therefore, we define a \emph{bundle} to be a nonempty set of transfers $\Theta=\{\theta_1,\dots,\theta_n\}$ such that  
$res(\theta_i)\neq{res}(\theta_j)$, for $i\neq{j}$, $i,j\in\{1,\dots,n\}$.
The set of all bundles for $(\mathbf{R},\mathbf{K})$ is denoted by ${BUN}(\mathbf{R},\mathbf{K})$.
A bundle $\Theta\in{BUN}(\mathbf{R},\mathbf{K})$ is \emph{jointly applicable in} 
$\varsigma\in{\cal{S}}(\mathbf{R},\mathbf{K})$ if, for each $\theta\in\Theta$, $\theta\in{APL}(\varsigma)$.
We then denote the set of bundles jointly applicable in $\varsigma$ by $JPL(\varsigma)$.
The \emph{joint application}, of $\Theta\in{BUN}(\mathbf{R},\mathbf{K})$ to 
$\varsigma\in{\cal{S}}(\mathbf{R},\mathbf{K})$, such that $\Theta\in{JPL}(\varsigma)$ 
(noted $jpl_{\Theta}(\varsigma)$), produces $\varsigma^\prime=(\varsigma\setminus\{[r,k_1]\mid(r,k_1,k_2)\in\Theta\})\cup\{[r,k_2]\mid(r,k_1,k_2)\in\Theta\}\in{\cal{S}}(\mathbf{R},\mathbf{K})$.

During the execution of a contract, resources which were originally not in the availability of any participating actor can become available, as if produced in compliance with the contract. 
For example, a certain required document can be printed and kept by a notar as legal registration of some act; or a deliverable, required in a tender contract, hence typically not existing before the contract was put in place, can be produced as required during the execution of the tender. 
On the other hand, some resource can become no longer available for any exchange, as for example a car crashed in an accident covered by an insurance policy. 
In order to model such cases, we enrich the set $\mathbf{K}$ with the special symbols $\top$ and $\bot$, 
deemed \emph{environment actors},
with the property that, for any resource $r\in\mathbf{R}$, both allocations $[r,\top]$ and $[r,\bot]$ are admissible, 
while no transfer of the form $(r,\bot,k)$ is admissible. 
We call the elements in ${\mathbf{K}}\setminus\{\top,\bot\}$ \emph{proper actors} and allocations of the form $[r,k]$, 
for some $k\not\in\{\top,\bot\}$, \emph{proper allocations}.

The introduction of $\top$ and $\bot$ also provides a way to formalise the occurrence of events relevant to a contract 
(e.g., the deadline for an instalment is reached, an accident report is submitted, a cargo is delivered) in a way formally identical to transfers of tokens.
In particular, we define a set $\mathbf{Ev}\subseteq{\mathbf{R}}$ of \emph{event tokens} (short, \emph{events}), such that, 
for $e\in\mathbf{Ev}$, any $\varsigma\in{\cal{S}}(\mathbf{R},\mathbf{K})$ can contain only one of $[e,\top]$, representing a situation where the event $e$ has not occurred yet, or $[e,\bot]$, representing a situation where $e$ has already occurred,
Then, the only possible transfer for $e$ is $(e,\top,\bot)$ modeling the 
%
%
occurrence of the event.
As a consequence, an event occurs atomically and can never occur again.  

A typical usage of transfers associated with event occurrences is to insert them into bundles, for transactions which must occur in correspondence with specific events. 
For example, in the loan contract, an instalment must be payed, and the corresponding receipt must be signed, with each 10th day of the month.


\section{Modeling contracts}\label{sec:automaton}
In this section, we show how a notion of contract, seen as a formal construct defining a structured collection of obligations, can be modeled in terms of resources, actors, and transfers, as introduced in Section~\ref{sec:contracts}. 

In particular, we show how specific bundles of transfers, whose application transforms states of affairs into states of affairs related to a universe ${\cal{S}}(\mathbf{R},\mathbf{K})$, correspond to transitions of a finite-state machine describing the possible executions of a contract on $(\mathbf{R},\mathbf{K})$. 
This constitutes the basis for relating sequences of transfers encoded in a ledger, as discussed in Section~\ref{sec:encoding}, to the evaluation of queries on the excution of the contract, as  discussed in Sections~\ref{sec:categoricalModel} and~\ref{sec:language}.

In essence, a contract must state, at least implicitly: 
\begin{inparaenum}[]
\item the conditions of its validity;
\item the acts through which the obligations of the contract are discharged and the conditions under which such acts can be carried out; and
\item the situations corresponding to the completion of the obligations set by the contract (we then say that the contract is \emph{honoured}) as well as those corresponding to a \emph{breach} of these obligations.
\end{inparaenum}  
The latter may be associated with repair or compensation actions, which may thus be seen as alternate ways for honouring the contract, or as constituting the definition of a different contract altogether.

We can then encode the legal form of a contract ${\cal{C}}$ in a \emph{legal contract automaton}, a finite state machine ${\cal{M}}^l_{\cal{C}}=(V^l,E^l,F,s^l,t^l,Act,TO,\lambda^l,\upsilon)$, where conditions correspond to states and the discharge of all of the obligations leading from a state to another corresponds to transitions. 
In particular: 
\begin{inparaenum}[{}]
	\item $V^l$ is the set of \emph{state nodes} (short, \emph{states});
	\item $E^l$ is the set of \emph{transition edges} (short, \emph{transitions}), with $s^l:E^l\rightarrow{V}^l$ the \emph{source map} and $t^l:E^l\rightarrow{V}^l$ the \emph{target map};
	\item $F\subseteq{V}^l$ is the set of \emph{final} states;
	\item $Act$ is the set of obligation-discharging \emph{actions}, with $TO\subseteq{Act}$ the set of \emph{timeout} actions; 
	\item $\lambda^l:E^l\rightarrow\wp(Act)$ is a function labeling each transition with the set of actions needed for its firing; and
	\item $\upsilon:F\rightarrow\{HON,BRC\}$ is a function mapping each final state into the corresponding outcome (honoured or breach).
\end{inparaenum}

We assume that if a timeout action $to$ is in $\lambda(\eta)$ for some $\eta\in{E}^l$, then $to$ is the only action in $\lambda(\eta)$, i.e., the latter is a singleton.
Moreover, we assume that $\C$ is such as to induce a strongly deterministic automaton ${\cal{M}}^l_{\cal{C}}$, 
in the sense that for $\eta_1,\eta_2\in{E}^l$ such that $s^l(\eta_1)=s^l(\eta_2)$, 
neither $\lambda^l(\eta_1)\subseteq\lambda^l(\eta_2)$, nor $\lambda^l(\eta_2)\subseteq\lambda^l(\eta_1)$ hold, while the case $\lambda^l(\eta_1)\cap\lambda^l(\eta_2)\neq\emptyset$ is admitted. 

The usual notion of trajectory between states, as given by the sequence of transitions leading from one to the other, can then be adapted to ${\cal{M}}^l_{\cal{C}}$.
The set of all trajectories in ${\cal{M}}^l_{\cal{C}}$ is denoted by $\Pi^l_{{\cal{C}}}$ and the set of all initial trajectories 
(i.e., originating in the initial state of ${\cal{M}}^l_{\cal{C}}$) is denoted by $\Pi^{l,in}_{{\cal{C}}}$. 

We introduce a sketch of a part of an insurance policy as a running example. 
%
%
\begin{example}\label{ex:insurance}
%
%
\emph{Fred} has activated a comprehensive insurance policy for his car with the company \emph{SURE!}, whereby a black box has been mounted on \emph{Fred}'s car. 
As a part of this policy, the coverage of expenses to repair car damages, even if due to natural causes, is provided, according to a certain procedure, defining a contract ${\cal{C}}_d$.
In particular, the procedure starts when a damage event occurs and the blackbox makes a report on it available to \emph{Fred}.  
Then \emph{Fred} has the possibility to file a claim for reimbursement on the \emph{SURE!} system and to upload the report on its server.
Then, the \emph{SURE!} system will issue an offer for reimbursement to \emph{Fred}, who might accept or reject it.
If \emph{Fred} communicates that he accepts the offer, the \emph{SURE!} system issues both a refund order and a communication that the policy premium is increased.
If \emph{Fred} communicates that he rejects the offer, no further action is needed.
Each act must be performed within some deadline. 
\end{example}

We sketch here the content of states and transitions defining ${\cal{M}}^l_{{\cal{C}}_d}$.

\begin{enumerate}[topsep=2pt,partopsep=1pt,itemsep=2pt,parsep=2pt]
	\item The procedure enters the \emph{Active} state when the damage event occurs AND the damage report is produced. A deadline is set to present the claim.
	\item The \emph{Claimed} state is reached when the claim is filed AND the report is acquired, before the deadline expires. At this stage, an obligation for \emph{SURE!} to make an offer becomes active, which must be discharged before a set deadline, on pain of violating the contract.
	\item The \emph{Offered} state is reached when the reimbursement offer is made, so that the obligation is discharged. An obligation for \emph{FRED} to respond to the offer is activated. This can be discharged by either accepting or rejecting the offer, but also by simply letting the deadline pass, after which the proposal becomes void (this would not constitute a violation).
	\item The \emph{Accepted} state is reached when \emph{FRED} accepts the reimbursement offer. An obligation for the \emph{SURE!} system to issue the refund AND to apply the raise to the premium is then activated. Also this obligation is under pain of violating the contract, if not discharged within the deadline.
	\item The \emph{Refunded} state is reached when the \emph{SURE!} system issues the refund AND applies the raise in the premium, before the deadline.
	\item The \emph{Rejected} state is reached when \emph{FRED} rejects the offer.
	\item Each missed deadline brings into a corresponding state. Deadlines missed by \emph{SURE!} represent a violation, while those missed by \emph{FRED} do not.
\end{enumerate}

Figure~\ref{fig:contract2} depicts the resulting legal contract automaton ${\cal{M}}^l_{{\cal{C}}_d}$. 
Each state from which no transition starts, namely \emph{Rejected}, \emph{Refunded}, and each of the four \emph{Out} states, is in $F$, as it represents a possible completion of the procedure. 
For $v\in{F}$, $v$ is depicted in green if $\upsilon(v)=HON$, in red if $\upsilon(v)=BRC$.

\begin{figure}[htb]
	\centering
	\includegraphics[width=.8\linewidth]{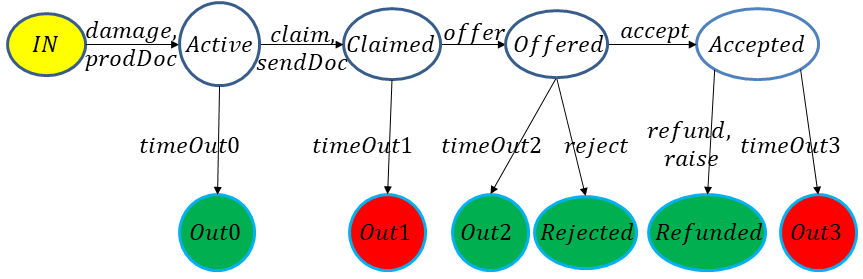}
	\caption{The finite state machine for the contract from Example~\ref{ex:insurance}.}\label{fig:contract2}
\end{figure}

It is to be noted that the procedure above is activated in the \emph{SURE!} system each time a damage is recorded, and proceeds according to the same steps independent of whether the customer is \emph{FRED} or any other. 
Moreover, the contract might even be realised as a service and offered to insurance companies which might customise it, e.g., by setting deadlines, setting maximum refunds, or defining premiums and their raise.
Without loss of generality, we can therefore consider that, for each insurer and customer, a fresh instance of the contract (and of the associated state machine) is always ready to control the next procedure, the state \emph{IN} in Figure~\ref{fig:contract2} representing the initial state for such an instance. 

We now proceed to modeling contracts in terms of the approach presented in Section~\ref{sec:contracts}. 
%
%
Indeed, leveraging the notions of \emph{resource}, \emph{actor}, \emph{transfer}, and \emph{bundle}, 
we identify all types of conditions in the contract with predicates on states of affairs and the required actions with transfers. 
The latter is to say that an obligation is considered to be discharged when there is evidence that some actor has provided some resource (which might represent a physical asset, access to services, or execution of some task) to some other actor, i.e., a transfer of the corresponding token has been recorded.
Where the contract requires a certain set of obligations to be completed before a different set comes into force, we model this in terms of a bundle of co-occurrent transfers.

Hence, the discharge of obligations in a legal contract automaton ${\cal{M}}^l_{\cal{C}}$ encoding a contract ${\cal{C}}$ realises some change in the states of affairs for an underlying system of tokens in a set $\mathbf{R}_{\cal{C}}$ 
(one for each piece of documentation mentioned in the contract) and actors in a set $\mathbf{K}_{\cal{C}}$ 
(one for each party involved in the contract).
In particular, given a state of affairs in ${\cal{S}}(\mathbf{R}_{\cal{C}},\mathbf{K}_{\cal{C}})$ holding in the initial state of ${\cal{M}}^l_{\cal{C}}$, each sequence of discharge actions (i.e., transfers) leads to the production of some state of affairs in ${\cal{S}}(\mathbf{R}_{\cal{C}},\mathbf{K}_{\cal{C}})$. 
It is now important to notice that, while a state transition is realised only when all the obligations needed are discharged 
(i.e., all the transfers in a bundle have been applied), it might be the case that other transfers, not composing a whole bundle, have been applied.  

As an example, $Out0$ is reached whenever the event $(timeOut0,\top,\bot)$ occurs before \emph{both} transfers corresponding to the discharge of the $claim$ and $sendDoc$ obligations have been applied. 
Hence, once ${\cal{M}}^l_{{\cal{C}}_d}$ reaches the state $Out0$, there are three possible corresponding states of affairs: 
one for not having applied either of $claim$ or $sendDoc$, and one each for having applied just one of them.

Moreover, not all states of affairs in ${\cal{S}}(\mathbf{R},\mathbf{K})$ can be mapped to some state of a legal contract automaton on $(\mathbf{R},\mathbf{K})$. 
For example, 
%
%
a state of affairs where a claim has been filed but the report of the damage event has not been made available, 
or, vice versa, such a report is available, but no claim has been filed (yet) for that accident, does not correspond to any state of ${\cal{M}}^l_{{\cal{C}}_d}$, as a consistent state will be reached only when both of the associated transfers will have been applied 
(or the $timeOut0$ event will have fired.). 

The model based on ${\cal{M}}^l_{\cal{C}}$ must therefore be integrated with a second finite state machine, univocally induced from ${\cal{M}}^l_{\cal{C}}$, 
the \emph{contract execution automaton} ${\cal{M}}^e_{\cal{C}}=(V^e,E^e,F,s^e,t^e,Act,TO,\lambda^e,\upsilon)$, 
defined on the same sets of actions and final states as ${\cal{M}}^l_{\cal{C}}$, which considers all possible sequences of actions (transfers) in a concrete execution of the contract, compatible with its legal definition. 
In ${\cal{M}}^e_{\cal{C}}$, transitions are labeled with single actions, realised individually, and not with a set of actions to be completely realised for the transition to occur.

In particular (and specifying only the components not inherited from ${\cal{M}}^l_{\cal{C}}$): 
\begin{inparaenum}[{}]
	\item $V^e\supseteq{V}^l$ is the set of \emph{states};
	\item $E^e$ is the set of \emph{transitions}, with $s^e:E^e\rightarrow{V}^e$ the \emph{source map} and $t^e:E^e\rightarrow{V}^e$ the \emph{target map}; and
	\item $\lambda^e:E^e\rightarrow{Act}$ is a function labeling each transition with the specific action needed for its firing.
\end{inparaenum}

The derivation of ${\cal{M}}^e_{\cal{C}}$ from ${\cal{M}}^l_{\cal{C}}$ 
proceeds according to Construction~\ref{con:execution}.

\begin{construction}\label{con:execution}
We isolate three components of the construction.
\begin{enumerate}[topsep=2pt,partopsep=1pt,itemsep=2pt,parsep=2pt]
	\item\label{it:intermediateStates} First, we include in $V^e$ the whole of $V^l$. 
	Then, for $\eta\in{E}^l$, we build all the ``intermediate'' states between $s^l(\eta)$ and $t^l(\eta)$, each corresponding to a nonempty proper subset of $\lambda^l(\eta)$. 
	We call $Lin(\eta)$ the resulting set of states.
	For example, given the states \emph{Active} and \emph{Claimed} of $V^l_d$, $Lin(\mbox{\emph{(Active,Claimed)}})$ contains the states \emph{Active/Claim} and \emph{Active/Sent}, corresponding to the partial realisations of the set of actions $\lambda^l(\mbox{\emph{(Active,Claimed)}})=\{\mbox{\emph{claim}},\mbox{\emph{sendDoc}}\}$.
	\item\label{it:intermediateTransitions} Then, for $\eta\in{E}^l$, we include in $E^e$ all the transitions connecting $s^l(\eta)$ to $t^l(\eta)$ through the states in $Lin(\eta)$. 
	That is to say that, for $\eta\in{E}^l$, and for $s_1,s_2\in{Lin}(\eta)\cup\{s^l(\eta),t^l(\eta)\}$, respectively associated with sets $\alpha_1,\alpha_2=\alpha_1\cup\{a\}$ for $a\in\lambda^l(\eta)\setminus\alpha_1$, $E^e$ will include a transition between $s_1$ and $s_2$ labeled with $a$. 
	Note that this also includes transitions leaving $s(\eta)$ (where $\alpha_1=\emptyset$) and transitions reaching $t(\eta)$ 
	(where $\alpha_2=\lambda^l(\eta)$). 
	For example, in $E^e_d$ a transition exists from \emph{Active} to each of \emph{Active/Claim} and \emph{Active/Sent}, labeled with the corresponding action, and from each of \emph{Active/Claim} and \emph{Active/Sent} to \emph{Claimed}, labeled with the action needed to complement the set in the original transition from \emph{Active} to \emph{Claimed} i.e., \emph{sendDoc} for \emph{Active/Claim} and \emph{claim} for \emph{Active/Sent}.
	\item The two items above are subsumed under a general construction, whereby $V^e$ and $E^e$ include all of the states and transitions needed to account for the possible interleaving of actions labeling different transitions in $E^l$ originating from a given state, until the completion of one of these sets 
	(no two transitions from the same state are labeled with the same set of actions), thus keeping track of the ``progress'' towards completion of each such set.
	To this end, for a state $v\in{V}^l$, we take all the transitions in ${E}^l$ with source in $v$, let this be $T(v)$, and add to $V^e$ a node for each set in $\wp(U(v))\setminus{Cmpl}(v)$, where $U(v)=\bigcup\{\lambda^l(\eta)\mid{\eta}\in{T}(v)\}$ and ${Cmpl}(v)=\{X\in{U}(v)\mid(\nexists\eta\in{E}^l)[\lambda^l(\eta)\subseteq{X}]\}$. 
	In other words, we consider all possible subsets in the union of all the sets of actions labeling these transitions, minus those subsets corresponding to the completion of one transition,	as only one such set can be completed. 
	Indeed, all and only the states in $V^l$ are taken to correspond to the completion of $\lambda^l(\eta)$ for any $\eta\in{T}(v)$.
	(This also means that once a state $v^\prime\in{V}^l$ is reached 
	(i.e., $\lambda^l(\eta)$ has been completed for some $\eta\in{T}(v),v^\prime=t^l(\eta)$), 
	all the actions representing progress towards other states become irrelevant, as a different set of intermediate states is constructed for transitions in $T(v^\prime)$). 
	Then, the set $E^e$ includes all the needed transitions between these states, one for each action representing an increment towards completion of $\lambda^l(\eta)$, for some $\eta\in{T}(v)$. 
	Transitions are only added towards reaching a state in $V^l$ which is the target for the corresponding transition in $E^l$, as defeasing of actions is not admitted.
	Note that by definition of timeout state, each such state is reached with just one action, so that a timeout would interrupt any possible trajectory in ${\cal{M}}^e_{{\cal{C}}}$ not yet landed in a state from $V^l$.
	For example, 
	%
	%
	irrespective of whether the \emph{timeOut0} event occurs in \emph{Active}, \emph{Active/Claim}, or \emph{Active/Sent}, it immediately completes the singleton set of actions needed to reach \emph{Out0}. 
	Hence, in addition to a transition from \emph{Active} to \emph{Out0}, generated in Item~\ref{it:intermediateTransitions}, 
	$E^e_d$ has a transition to \emph{Out0} from each of \emph{Active/Claim} and \emph{Active/Sent}, both of them labeled \emph{timeOut0}.
\end{enumerate}
\end{construction} 
%
Similarly to the case for ${\cal{M}}^l_{\cal{C}}$, we denote the set of trajectories in ${\cal{M}}^e_{\cal{C}}$ by $\Pi^e_{{\cal{C}}}$ and the set of initial trajectories by $\Pi^{e,in}_{{\cal{C}}}$. 
The difference between considering all trajectories in $\Pi^{e,in}_{\cal{C}}$ or only those passing through only the intermediate states built in $Lin(\eta)$ for some $\eta\in{E}^l$,
(which correspond to a trajectory in $\Pi^l_{\cal{C}}$) will play an essential role in the definition of the 
%
%
logics in Section~\ref{sec:categoricalModel}.


Another caveat is in order. 
In principle, a state of affairs, even if containing allocations for all of the resources and actors involved in a contract, might include allocations for a set of resources concerning other aspects of their rapport. 

For example, and extending Example~\ref{ex:insurance}, some tokens in the complete model for a car insurance policy between an \emph{insurer} company and a \emph{customer} might refer to topics such as \emph{coverage of theft},\emph{ extension of civil liability}, 
\emph{deadlines for payments,} etc, not relevant to the procedures for damage management. 
Hence, 
%
%
%
an automaton ${\cal{M}}_{{\cal{C}}_x}$ (either legal or execution) might simply model some specific section ${\cal{C}}_x$ in a wider contract ${\cal{C}}$, relative to a universe
of states of affairs ${\cal{S}}(\mathbf{R}_{\cal{C}},\mathbf{K}_{\cal{C}})$.
Then, we should regard states in ${\cal{M}}_{{\cal{C}}_x}$ as inducing some relation on (some subset of) 
%
%
${\cal{S}}(\mathbf{R}_{\cal{C}},\mathbf{K}_{\cal{C}})$.
A compositional view of contracts can then ensue, resulting in progressive refinements of this relation.

In general, for a contract ${\cal{C}}$, we will identify the subset ${\cal{S}}_{\cal{C}}\subseteq{\cal{S}}(\mathbf{R}_{\cal{C}},\mathbf{K}_{\cal{C}})$, of states of affairs
\emph{consistent with} ${\cal{M}}^l_{\cal{C}}$, in the sense that they can occur as the combined effect of a sequence of transfers which can occur as a possible execution of the obligations in the contract, according to ${\cal{M}}^e_{\cal{C}}$.

The translation of ${\cal{M}}^l_{\cal{C}}$ in terms of a $(\mathbf{R}_{\cal{C}},\mathbf{K}_{\cal{C}})$ system 
%
%
proceeds as follows.

\begin{enumerate}[topsep=2pt,partopsep=1pt,itemsep=2pt,parsep=2pt]
	\item Each party involved in ${\cal{C}}$  is modelled as an actor $k\in\mathbf{K}_{\cal{C}}$.
	\item Each resource to be produced or transfered in order to discharge an obligation in ${\cal{C}}$ is modelled as a token
	 $r\in\mathbf{R}_{\cal{C}}$. In particular, the production of a document corresponds to a transfer from $\top$ to a proper actor.
	\item Each accident to be documented in ${\cal{C}}$ is modelled as an event $e\in\mathbf{Ev}_{\cal{C}}\subseteq\mathbf{R}_{\cal{C}}$.
	The expiration of a deadline is also modelled as an event. Event occurrences are therefore modelled as transfers of the form $(e,\top,\bot)$.
	\item Each action discharging an obligation (i.e., in $Act$) is modelled as a transfer, in a set $\mbox{TRA}_{\cal{C}}$, of the resource associated with that obligation.
	\item Each set of actions collectively ensuring a transition between states in ${\cal{M}}^l_{\cal{C}}$ is modelled as a bundle, in a set $\mbox{BUN}_{\cal{C}}$, of the corresponding transfers.
\end{enumerate}

We can now define the system of resources and actors modeling a contract.

\begin{definition}[Resource-based contract model]
Let ${\cal{C}}$ be a contract and let ${\cal{M}}^l_{\cal{C}}=(V^l,E^l,F,s^l,t^l,Act,\lambda^l,\upsilon)$ be the associated contract automaton.
Then a \emph{resource-based contract model} is a tuple $\mathbf{RS}_{\cal{C}}=({\cal{S}}_{\cal{C}},\mbox{BUN}_{\cal{C}},\mbox{TRA}_{\cal{C}})$, together with three mappings:
$\gamma:\Pi_{{\cal{C}}}\rightarrow\wp({\cal{S}}_{\cal{C}})$, $\beta:E^l\rightarrow\mbox{BUN}_{\cal{C}}$, and $\rho:Act\rightarrow\mbox{TRA}_{\cal{C}}$, collectively enjoying the following properties:

\begin{enumerate}[topsep=2pt,partopsep=1pt,itemsep=2pt,parsep=2pt]
	\item For $\tau{\in}\Pi_{{\cal{C}}},\eta{\in}{E}^l,\Theta{=}\beta(\eta)$: $\langle\tau\cdot\eta\rangle{\in}\Pi_{{\cal{C}}}\Rightarrow(\forall\varsigma{\in}\gamma(\tau))[jpl_\Theta(\varsigma){\in}\gamma(\langle\tau\cdot\eta\rangle)]$.
	\item For $\eta\in{E}^l$: $\beta(\eta)=\{\rho(act)\mid{act}\in\lambda(\eta)\}$\footnote{That is, $\beta$ is completely consistent with $\lambda^l$.}.
\end{enumerate}
\end{definition}

We can now show the constrution of the legal contract automaton ${\cal{M}}^l_{{\cal{C}}_d}$ for the contract ${\cal{C}}_d$ of 
Example~\ref{ex:insurance}, managing damage events 
(the contract execution automaton ${\cal{M}}^e_{{\cal{C}}_d}$ is induced according to Construction~\ref{con:execution}). 
Since each arrow in $E^l$ for ${\cal{M}}^l_{{\cal{C}}_d}$ is unique to an ordered pair of states, we identify here an edge in $E^l$ with the corresponding pair. 
Moreover, since, for each state in $V^l$, there is a unique trajectory leading to it, we identify a state $v$ and the corresponding  trajectory, denoted by $\overline{v}$, in the definition of $\gamma$.

\begin{example}\label{ex:remo}	
With reference to Example~\ref{ex:insurance}, we define:
	\begin{itemize}[topsep=2pt,partopsep=1pt,itemsep=2pt,parsep=2pt]
		\item $\mathbf{R}_{{\cal{C}}_d}=\{oldPrem,claim,damageEv,damageDoc,\mbox{offer},{reject},{accept},raise$, 
		$\mbox{refund},out0,out1,out2,out3\}$.
		\item $\mathbf{K}_{{\cal{C}}_d}=\{customer,insurer,\bot,\top\}$.
		\item The function $\gamma$ is defined as follows\footnote{As is customary, we will write a singleton $\{\phi\}$ as simply $\phi$, for a set $\phi$.}:
		\begin{itemize}[topsep=2pt,partopsep=1pt,itemsep=2pt,parsep=2pt]
			\item $\gamma(\overline{In})=\{[oldPrem,customer]\}\cup\{[r,\top]\mid{r}\in\mathbf{R}_{{\cal{C}}_d}\setminus\{oldPrem\}\}$.
			\item $\gamma(\overline{Active})=\{[oldPrem,customer],[damageDoc,customer],[damageEv$, 
			$\bot]\}\cup\{[r,\top]\mid{r}\in\mathbf{R}_{{\cal{C}}_d}\setminus\{oldPrem,damageEv,damageDoc\}\}$.			
			\item $\gamma(\overline{Claimed})=\{[claim,insurer],[damageDoc,insurer],[damageEv,\bot]$, $[oldPrem,customer]\}\cup\{[r,\top]\mid{r}\in\mathbf{R}_{{\cal{C}}_d}\setminus\{claim,damageEv$, $oldPrem,damageDoc\}\}$;
			\item $\gamma(\overline{\mbox{Offered}})=\{[\mbox{offer},customer],[damageEv,\bot],[damageDoc,insurer]$,
			$[oldPrem,customer],[claim,insurer]\}\cup\{[r,\top]\mid{r}\in\mathbf{R}_{{\cal{C}}_d}\setminus\{\mbox{offer}$, 
			$damageEv,damageDoc,oldPrem,claim\}\}$;
			\item $\gamma(\overline{Accepted})=\{[accept,insurer],[\mbox{offer},customer],[oldPrem,customer]$,  $[damageEv,\bot],[claim,insurer],[damageDoc,insurer]\}\cup\{[r,\top]\mid{r}\in\mathbf{R}_{{\cal{C}}_d}\setminus\{claim,accept,damageEv,damageDoc,\mbox{offer},oldPrem\}\}$.
			\item $\gamma(\overline{\mbox{Refunded}})=\{[accept,insurer],[raise,customer],[damageEv,\bot]$, 
			$[damageDoc,insurer],[\mbox{refund},customer],[claim,insurer],[oldPrem$, 
			$\bot]\},\cup\{[r,\top]\mid{r}\in\{reject,out0,out1,out2,out3\}$;
			\item $\gamma(\overline{Rejected})=\{[damageEv,\bot],[damageDoc,insurer],[claim,insurer]$, 
			$[\mbox{offer},customer],[reject,insurer]\}\cup\{[r,\top]\mid{r}\in\mathbf{R}_{{\cal{C}}_d}\setminus\{claim,reject$, 
			$\mbox{offer},damageDoc,damageEv\}\}$;
			\item $\gamma(\overline{Out0})=\{\varsigma{\in}{\cal{S}}(\mathbf{R}_{{\cal{C}}_d},\mathbf{K}_{{\cal{C}}_d})\mid[out0,\bot]{\in}\varsigma,[outX,\top]{\in}\varsigma\mbox{ for }X{\in}\{1,2,3\}\}$.
			\item $\gamma(\overline{Out1})=\{\gamma{\in}{\cal{S}}(\mathbf{R}_{{\cal{C}}_d},\mathbf{K}_{{\cal{C}}_d})\mid[out1,\bot]{\in}\gamma,[outX,\top]{\in}\varsigma\mbox{ for }X{\in}\{0,2,3\}\}$.
			\item $\gamma(\overline{Out2})=\{\gamma{\in}{\cal{S}}(\mathbf{R}_{{\cal{C}}_d},\mathbf{K}_{{\cal{C}}_d})\mid[out2,\bot]{\in}\gamma,[outX,\top]{\in}\varsigma\mbox{ for }X{\in}\{0,1,3\}\}$;
			\item $\gamma(\overline{Out3})=\{\gamma{\in}{\cal{S}}(\mathbf{R}_{{\cal{C}}_d},\mathbf{K}_{{\cal{C}}_d})\mid[out3,\bot]{\in}\gamma,[outX,\top]{\in}\varsigma\mbox{ for }X{\in}\{0,1,2\}\}$.
			\end{itemize}
		The set $BUN_{{\cal{C}}_d}$ and the function $\beta$ are jointly defined as follows, implicitly defining also the set $TRA_{{\cal{C}}_d}$ and the $\rho$ function:
		\begin{itemize}[topsep=2pt,partopsep=1pt,itemsep=2pt,parsep=2pt]
			\item $\beta(In,Active)=\{(damageEv,\top,\bot),(damageDoc,\top,customer)\}$ 
			\item $\beta(Active,Claimed)=\{(claim,\top,insurer),(damageDoc,customer$, $insurer)\}$, 
			\item $\beta(Claimed,\mbox{Offered})=\{(\mbox{offer},\top,customer)\}$, 
			\item $\beta(\mbox{Offered},Accepted)=\{(accept,\top,insurer)\}$,
			\item $\beta(Accepted,\mbox{Refunded})=\{(\mbox{refund},\top,customer),(oldPrem,customer$, $\bot),(raise,\top,customer)\}$,
			\item $\beta(\mbox{Offered},Rejected)=\{(reject,\top,insurer)\}$,
			\item $\beta(Active,Out0)=\{(out0,\top,\bot)\}$,
			\item $\beta(Active,Out1)=\{(out1,\top,\bot)\}$,
			\item $\beta(Active,Out2)=\{(out2,\top,\bot)\}$,
			\item $\beta(Active,Out3)=\{(out3,\top,\bot)\}$,
		\end{itemize}  
	\end{itemize}  
\end{example}

By construction of $\mathbf{RS}_{\cal{C}}$, one can derive the sets of sequences of bundles and of transfers leading from any given state of affairs in $\mathbf{RS}_{\cal{C}}$ to any other state of affairs in $\mathbf{RS}_{\cal{C}}$ reachable from the first, following a trajectory in ${\cal{M}}^l_{\cal{C}}$.

The functions $\beta$ and $\rho$ then induce the two notions of \emph{bundle} and 
\emph{transfer} \emph{trajectory-labeling}, for trajectories in ${\cal{M}}^l_{\cal{C}}$ and ${\cal{M}}^e_{\cal{C}}$, respectively. 

\begin{definition}[Trajectory labelings]\label{def:transferLabeling}
	Let ${\cal{C}}$ be a contract, let ${\cal{M}}^l_{\cal{C}}$ and ${\cal{M}}^e_{\cal{C}}$ be the associated legal contract  and contract execution automata, and let $\tau^l=\langle{t}^l_1\cdots{t}^l_n\rangle\in\Pi^l_{{\cal{C}}}$ and $\tau^e=\langle{t}^e_1\cdots{t}^e_m\rangle\in\Pi^e_{{\cal{C}}}$, be such that all and only the states in $V^l$ reached with $\tau^e$ are also reached with $\tau^l$, in the same order. In this case, we say that $\tau^e$ is an \emph{unfolding} of $\tau^l$, and we define:
	\begin{itemize}[topsep=2pt,partopsep=1pt,itemsep=2pt,parsep=2pt]
		\item $B^\beta_{{\cal{C}}}(\tau^l)=\langle\beta(t^l_1)\cdots\beta(t^l_n)\rangle$ to be the \emph{bundle trajectory-labeling of} $\tau^l$;
		\item $L^\rho_{{\cal{C}}}(\tau^e)=\langle\rho(t^e_1)\cdots\rho(t^e_n)\rangle$,
		 to be the \emph{transfer trajectory-labeling} of $\tau^e$.
	\end{itemize}
If $L^\rho_{{\cal{C}}}(\tau^e)$ is such that only states in $\bigcup_{\eta\in{E}^l}{Lin}(\eta)$ 
(see Construction~\ref{con:execution}) are visited, then we say that $L^\rho_{{\cal{C}}}(\tau^e)$ is a \emph{linearisation} of $B^\beta_{{\cal{C}}}(\tau^l)$.
\end{definition}

The set of all bundle trajectory-labelings for trajectories in $\Pi^l_{{\cal{C}}}$ is denoted by $\mathbf{B}^\beta_{{\cal{C}}}$ and the set of all bundle trajectory-labelings for initial trajectories in $\Pi^{l,in}_{{\cal{C}}}$ is denoted by $\mathbf{B}^{\beta,in}_{\cal{C}}$.
Analogously, the set of all transfer trajectory-labelings for trajectories in $\Pi^e_{{\cal{C}}}$ is denoted by
$\mathbf{L}^\rho_{{\cal{C}}}$ and the set of all trajectory transfer labelings for initial trajectories in $\Pi^{e,in}_{\cal{C}}$ is denoted by $\mathbf{L}^{\rho,in}_{\cal{C}}$.
A prefix-induced partial order is therefore defined on both sets of labelings.
	
	\begin{enumerate}[topsep=2pt,partopsep=1pt,itemsep=2pt,parsep=2pt]
	\item Given two trajectories $\tau^l_1,\tau^l_2\in\Pi^l_{{\cal{C}}}$, $\tau^l_1{\le^l_\C}\tau^l_2$ \emph{iff} $\tau^l_1{\in}{PREF}(\tau^l_2)$
	\item Given two bundle trajectory-labelings $\beta_1=B^\beta_{{\cal{C}}}(\tau^l_1)$ and $\beta_2=B^\beta_{{\cal{C}}}(\tau^l_2)$, 	$\beta_1{\le_\beta}\beta_2$ \emph{iff} $\beta_1{\in}{PREF}(\beta_2)$ \emph{iff} $\tau^l_1{\le^l_\C}\tau^l_2$.
	\item Given two trajectories $\tau^e_1,\tau^e_2\in\Pi^e_{{\cal{C}}}$, $\tau^e_1{\le^e_\C}\tau^e_2$ \emph{iff} $\tau^e_1{\in}{PREF}(\tau^e_2)$.
	\item Given two transfer trajectory-labelings $\rho_1=L^\rho_{{\cal{C}}}(\tau^e_1)$ and $\rho_2=L^\rho_{{\cal{C}}}(\tau^e_3)$, $\rho_1{\le_\rho}\rho_2$ \emph{iff} $\rho_1{\in}{PREF}(\rho_2)$ \emph{iff} $\tau^e_1{\le^e_\C}\tau^e_2$.
	\end{enumerate}

\begin{remark}\label{rem:naming}
	Strictly speaking, bundles and transfers are complex structures, so that in principle we should distinguish between their definition in terms of declarative specifications of behaviours modifying a state of affairs and their unique names to be used in labeling. 
	This could be achieved by associating with each transfer $\theta$ its unique name $\overline{\theta}$ and with each bundle $\Theta$ its unique name $\overline{\Theta}$. 
	For the sake of simplicity we do not introduce this distinction here, relying on the context to clarify whether we are referring to names or to specifications.
\end{remark}

\section{Encoding transfers on ledgers}\label{sec:encoding}
In this section, we discuss how the model of contracts presented in Section~\ref{sec:automaton}, which views them as defining admissible evolutions of the state of affairs of a system of resources and actors, lends itself to the recording of sequences of contract-related actions on a ledger, making it amenable to forms of auditing on their compliance with the constraints set by the contract. As a consequence,  we do not simply deal with the ``current'' set of allocations, as maintained in traditional databases, but we aim at reconstructing the whole ``history'' of transfers involving resources and actors pertaining to a contract. 

Before moving on, it is important to clarify that we only consider the recording of actions on a ledger, while we are agnostic regarding the deployment of the contract which can itself reside on the ledger, as in some blockchains, rather than on a network node, as in other cases. In either case case the contract and its actions are at different levels and the focus here is on actions.

The \emph{allocation history} of a resource $r\in\mathbf{R}$ can then be reconstructed by considering, in the sequence of encodings of applied transfers, those for which $res(\theta)=r$.
We denote the set of sequences of (encodings of) transfers (irrespective of the contract for which it has been executed) in the ledger by $LT$. 
Each sequence $\sigma\in{LT}$ is called a \emph{ledger state}.
Note that $LT$ is \emph{prefix-complete} with respect to the standard prefix-order $\le$
(i.e., for any $\sigma\in{LT}$, $PREF(\sigma)\subseteq{LT}$).
%

In a centralized ledger, the recording of transfers has a natural correspondence with contract execution. 
In a distributed environment, however, as they can originate in different nodes, we cannot assume that a log state records transfers in the exact order in which they occurred.
We can assume, however, that a \emph{resource-safeness} (see Definition~\ref{def:properties}) property holds in every encoding sequence registered in the ledger.
That is, for $r\in\mathbf{R}$, a new transfer of $r$ is not logged if the previous transfer of $r$ has not yet been transferred to the ledger.
This property can be verified through a simple check during coding.
 
Besides resource-safeness, we consider other properties, 
providing the basis for some form of conformance-checking on the traces of transfers associated with a contract. 
Hence, compliance to a \emph{wallet-safeness} property amounts to requiring that actors can only use resources they are entitled to 
(e.g., no form of double spending can be encoded in the ledger); 
compliance to a \emph{bundle-safeness} property provides some form of transactionality 
(if we observe a transfer from a bundle $\Theta$ for some $r\in\mathbf{R}_{\cal{C}}$, then the next transfer for $r$ can only appear after $\Theta$ is completed); 
and compliance to a \emph{contract-safeness} property, with respect to a contract ${\cal{C}}$, means that the sequence of transfers encoded in the ledger is consistent with both partial orders, $\le_\beta$ and $\le_\rho$, induced by ${\cal{M}}^l_{\cal{C}}$ and ${\cal{M}}^e_{\cal{C}}$.

Definition~\ref{def:properties} provides a formal account of these properties. 
From now on, we introduce a notion of \emph{encoding} of a transfer $\theta=(r,k_1,k_2)$ on the ledger, noted $\varrho(\theta)$.
An encoding is a construct maintaining information about $r,k_1$, and $k_2$, together with some metadata, such as a timestamp for its creation time, 
the identifier of some validation authority, or of the contract to which it refers.

\begin{definition}[Properties of a ledger state]\label{def:properties}
Let $\mathbf{R}$ be a set of resources and let $\sigma=\langle\varrho({\theta}_1)\cdots\varrho({\theta}_k)\rangle$ be a ledger state, with $res(\theta_i)\in\mathbf{R}$ for $i\in\{1,\dots,k\}$. 
Then, we define:

\begin{itemize}[topsep=2pt,partopsep=1pt,itemsep=2pt,parsep=2pt]
	\item For $r\in\mathbf{R}$, $\sigma$ is $r$-\emph{safe} if: for any $j>i\in[k]$,  
	such that $\theta_i=(r,k_1,k_2),\theta_j=(r,k_3,k_4)$, 
	for some $k_3\neq{k_2}$, 
	$\sigma$ contains a sub-sequence $\langle\varrho(\theta_{l_1})\cdots\varrho(\theta_{l_m})\rangle$, $l_1>i,l_m\le{j}$, 
	such that, for $s\in\{1,\dots,m{-}1\}$, $\theta_{l_s}=(r,k_{l_s},k_{l_{s{+}1}})$, with $k_{l_1}=k_2,k_{l_{m}}=k_3$.

\item For a contract ${\cal{C}}$ (so that $\mathbf{R}=\mathbf{R}_{\cal{C}}$):
\begin{itemize}
	\item $\sigma$ is ${\cal{C}}$-\emph{wallet safe} if it is $r$-safe for any $r\in\mathbf{R}_{\cal{C}}$;
	
	\item $\sigma$ is ${\cal{M}}^l_{\cal{C}}$-\emph{bundle safe} if it is ${\cal{C}}$-wallet safe and, 
	for $\Theta\in\beta(E^l),\theta\in\Theta$, $\varrho(\theta)$ appearing in $\sigma$: no encoding $\varrho(\theta^\prime)$, with $res(\theta^\prime){=}res(\theta)$, appears in $\sigma$ before all transfers in $\Theta$ have been encoded in $\sigma$.
	
	\item $\sigma$ is ${\cal{M}}^l_{\cal{C}}$-\emph{contract safe} if it is ${\cal{M}}^l_{\cal{C}}$-bundle safe and it is a prefix of some transfer trajectory-labeling of ${\cal{M}}^l_{\cal{C}}$.
\end{itemize}
\end{itemize}
\end{definition}

\begin{proposition}\label{prop:properties}
Each of the properties in Definition~\ref{def:properties} is decidable.
\end{proposition}

\begin{proof}[Sketch]
It is easy to see that by inspecting the sequence $\sigma_\C$, of transfers relative to resources in some contract $\C$, extracted from a ledger state $\sigma$, and comparing it with the sequences of transfers in $\mathbf{L}^\rho_{{\cal{C}}}$ 
(using the induced contract execution automaton ${\cal{M}}^e_{\cal{C}})$, each of the considered properties can be assessed. 
In particular, a ledger state is ${\cal{M}}_{\cal{C}}$-contract safe \emph{iff} its extracted sequence for resources in $\mathbf{R}_{\cal{C}}$ is in $\mathbf{L}^{\rho,in}_{{\cal{C}}}$. 
Since a ledger state is constituted of a finite number of transfers,
the extracted sequence $\sigma_\C$ needs to be compared with a finite number of transfer trajectory-labelings\footnote{Remember that $\mathbf{L}^{\rho,in}_{{\cal{C}}}$ is prefix-complete.} up to the length $|\sigma_\C|$.
\end{proof}

The proof is then immediate for Corollary~\ref{cor:safePrefix}.

\begin{corollary}\label{cor:safePrefix}
All prefixes of a ${\cal{M}}^l_{\cal{C}}$-contract safe state are ${\cal{M}}^l_{\cal{C}}$-contract safe.
\end{corollary}

In order to define the set of sequences of transfers encoded in a ledger, corresponding to possible contract executions, we have to go back to the considerations in Section~\ref{sec:automaton} on the possibility that some bundles do not get completed, so that some of the transfers in them appear in a labeling in $\mathbf{L}^{\rho,in}_{\cal{C}}$, but not in any linearisation of a labeling in  $\mathbf{B}^{\beta,in}_{\cal{C}}$.  

In particular, resuming the arguments in Construction~\ref{con:execution}, for each state $v\in{V}^l$, 
the determinism of ${\cal{M}}^l_{\cal{C}}$ allows us to establish a bijection between the set $T(v)$, 
of transitions leaving $v$, and the set
$Bun(v)=\{\beta(\eta)\mid\eta\in{T}(v)\}$ of bundles labeling these transitions.
%
%
Then $\mbox{\emph{Trf}}(v)=\bigcup{Bun}(v)$ is the set of transfers in all the bundles labeling transitions in $T(v)$. 
Let $\sigma$ be a ledger state such that its last encoded transfer ``completes" a bundle in $Bun(v)$, leading to a state $v^\prime\in{V}^l$. No bundle which has not been completed in $\sigma$ can then be completed in any prolongation of $\sigma$.
As a consequence, $\sigma$ comprises a sparse subsequence $\sigma^\prime$ of (encodings of) transfers which are parts of the transfer trajectory-labeling for a trajectory $\tau\in\Pi^e_{{\cal{C}}}$ leading to $v^\prime$, and a sparse subsequence $\sigma^{\prime\prime}$ of  (encodings of) transfers not on this trajectory, and consequently, not on any trajectory which is a prolongation of $\tau$. 
We say that  $\sigma^\prime$ is ``useful" and that  $\sigma^{\prime\prime}$  is ``useless". 
The reasoning can be extended to sequences which proceed beyond the last completed bundle, i.e., after reaching a state $v\in{V}^l$. 
Let $\sigma$ be one such sequence and let $\sigma_v$ the suffix of $\sigma$ following the completion of a bundle leading to $v$.
Then, if a prolongation of $\sigma_v$, say $\sigma^\prime_v$, leads to a new state $v_2$, the transfers in $\beta((v,v_2))$ will be incorporated in the useful part of $\langle\sigma\cdot\sigma^\prime_v\rangle$, while the remaining transfers in $\sigma^\prime_v$ will be incorporated in its useless part.

\section{An algebraic model of contracts and its logic}\label{sec:categoricalModel}
We can now go and resume our program of constructing a suitable logic for ledgers and contract automata based on 
their associated algebraic structures.
The presented theory is inspired by, but not immediately reducible to, a general theory introduced in categorical terms in~\cite{KL99}, used in~\cite{BGKL18} to provide a logic for (possibly nondeterministic) processes, and adapted to the study of contracts in~\cite{BGKL19}. 
In fact, the idea originated from the natural association of a Heyting first-order logic to a category of generalised labeled trees, proven to be a Heyting category and extended with several modal/temporal operators in~\cite{BGKL18}. 
There, the labeling of trees via a meet-semilattice was a crucial device, since we had to deal with a non-deterministic situation, 
which we modeled by allowing two different paths to be labeled via the same element in the meet-semilattice.

When, as here, one has a deterministic situation 
(i.e., one path, either in ${\cal{M}}^l_{\cal{C}}$ or in ${\cal{M}}^e_{\cal{C}}$ can only have a unique label), 
the labeling machinery can be dropped to consider the meet-semilattice itself as a tree, and the theory can be expressed in (po)set-theoretical terms. 
Indeed, the structure associated with the set of subtrees of a tree $\X$ is now that of a boolean algebra, so that, by using subtrees as interpretations of formulas,a boolean logic is obtained, (extendable with modal/temporal operators, see~\cite{BGKL18}), in analogy with the standard interpretation of formulas in terms of subsets of some universe of discourse. 

In our approach, this kind of structure can be associated both with a ledger used to register (encodings of) transfers performed 
under the constraints set by (possibly more than) one pair of contract automata, the legal and the execution one, and with the behaviour of the contract automata themselves (or more of them).
As a consequence, we are interested in two (or more) meet-semilattices at the same time, one derived from the sequence of transfers recorded in the ledger, the other ones derived from the admissible initial sequences of bundles (transfers) in a bundle (transfer) trajectory-labeling for some contract (or contracts) whose executions are recorded in the ledger.
 
 Without loss of generality, we restrict ourselves to one ledger and one contract with its two automata.
 The involved meet-semilattices correspond to sorts in a category canonically associated with a classical many-sorted logic, namely,  
 the category $\mathbf{TSL}$ of trees derived from meet-semilattices, with monotonic functions as morphisms.
 Indeed, for every object $\X$ of $\mathbf{TSL}$, $Sub(\X)$ is a Boolean algebra, and 
 morphisms allow a canonical definition of quantifiers, as shown later. 
 The logic associated with $\mathbf{TSL}$ will be expanded with modal/temporal operators, again defined canonically from the order relation on paths. 
 
Given a finite alphabet  of transactions $T$, the monoid $T^*$, of freely-generated sequences of transactions, is canonically defined. 
This is a meet-semilattice as well as a tree (rooted in the empty sequence $\epsilon$)  in our sense; all the prefixes of elements in $T^*$ are still in $T^*$. 

However, due to the condition of resource-safeness, 
not all possible sequences in $T^*$ can actually constitute the record, in a ledger, of a sequence of transfers.
Hence, the tree of possible ledger states
forms a prefix-closed proper subtree of $T^*$, $LT$ (as per the discussion in Section~\ref{sec:encoding}).  
When a new transfer is registered on the ledger, 
this causes the selection in $LT$ of all and only those paths which present that transfer at the corresponding step,
and which are like-wise consistent at all previous steps,
thus eliminating all the paths which are not its prolongations from the possible evolutions of the ledger. 
This is precisely the novelty in this approach: that we consider a special subfamily of $Subtree(LT)$, namely what we call the \emph{evolutions} of the ledger.
If we take $\E_0=LT$ to define the set of ledger states which are still reachable at time $0$, 
by repeating this operation any number of times $n$, we obtain a sequence of 
(\emph{instantaneous}) \emph{evolutions}
$\langle\E_0,\E_1,\dots,\E_n\rangle$ (a subset of $Subtree(LT)$), with $\E_{i+1}\subseteq{\E}_i$, for $i\in\{0,\dots,n{-}1\}$. 

\begin{remark}
	We observe that, in general, one has  $\E_{i+1}\subsetneq{\E}_i$ 
	The case $\E_{i+1}={\E}_i$ occurs only for $i=n{-}1$, in a situation where, after applying all transfers in $\E_i$, the only allocation of the form $(r,k)$, for $k{\neq}\bot$, in the resulting state of affairs is such that $r\in\mathbf{Ev}$, so that the only possible transfer left at step $i{+}1$ is $(r,\top,\bot)$. 
\end{remark}

In this perspective, the 
overall
ledger evolution is modeled as a sequence of trees on $(T^*,\leq,\wedge,\lambda)$, each containing the following one in the sequence:
after having established a finite set of transfer records, 
only those paths 
which are prolongations of it
are selected to generate the prefix-complete tree $\E_t$, representing the instantaneous evolution at time $t$.
Figure~\ref{fig:evolution} presents an intuitive representation of this sort of pruning, occurring from one step to the next.

\begin{figure}[htb]
	\centering
\includegraphics[width=9cm]{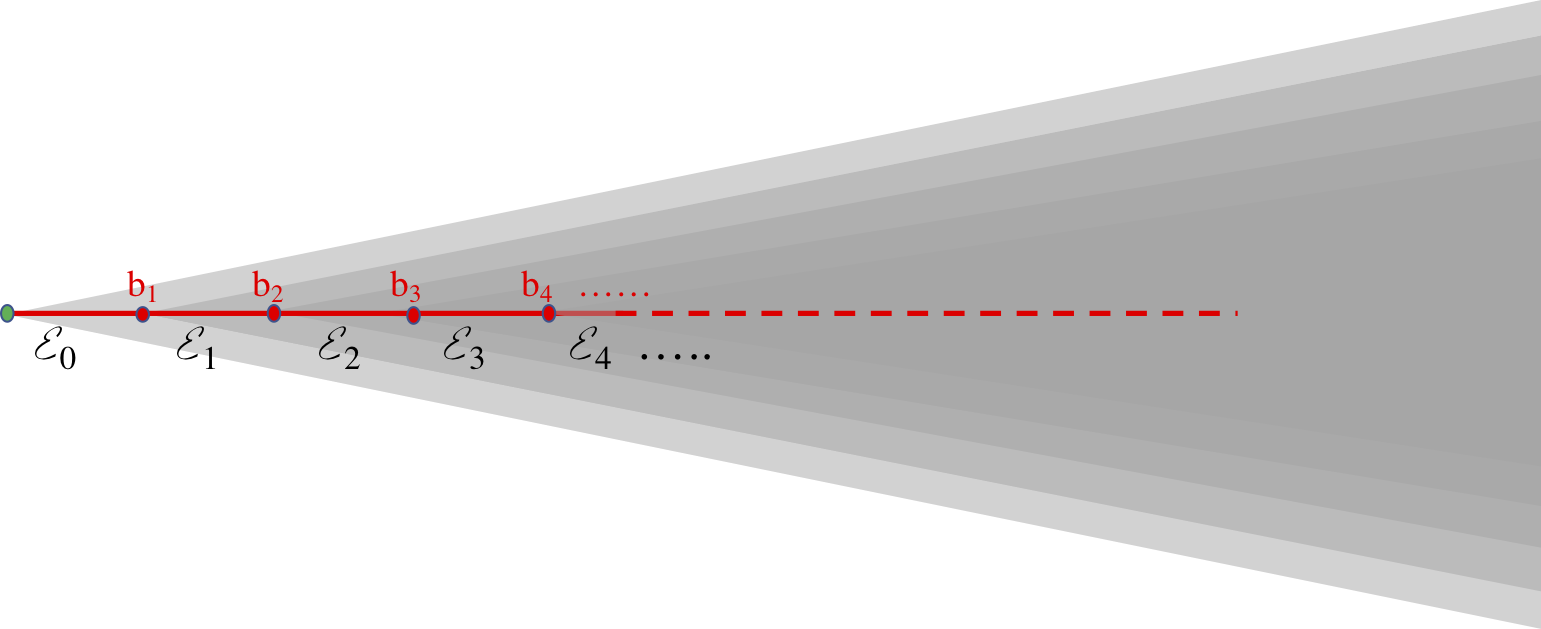}
\caption{An evolving ledger.}\label{fig:evolution}
\end{figure}
 
The resulting chain of instantaneous evolutions will in turn represent the global evolution.
Every instantaneous evolution $\E_k$ contains an \emph{initial chain} of paths
(the \emph{established} part coloured in red in Figure~\ref{fig:evolution}), of length $k$.

On the other hand,  we can easily see, on the basis of what discussed in Section~\ref{sec:automaton}, that both the set of all bundle initial trajectory-labelings in $\Pi^{l,in}_C$, $\mathbf{B}^{\beta,in}_C$, and the set of all transfer initial trajectory-labelings in $\Pi^{e,in}_C$, $\mathbf{L}^{\rho,in}_C$, with their partial orderings, are semilattices (and trees).

We will now define an \emph{occurring function} for each of them: 
namely $\nu^l:LT\rightarrow \mathbf{B}^{\beta,in}_C$ 
(resp. $\nu^e:LT\rightarrow \mathbf{L}^{\rho,in}_C$).
Intuitively, given a contract $\C$, an occurring function extracts, from a given initial sequence $\sigma$ of transfers recorded on the ledger, the subsequence $\sigma^\prime$ of those 
%
%
relative to  $\C$, 
and associates with it 
%
%
the maximal initial trajectory of states in ${\cal{M}}^l_{\C}$ (resp. ${\cal{M}}^e_{\C}$) reached through $\sigma^\prime$. 

\begin{definition}[Occurring maps]\label{def:occurringFunction}
	Let $\C$ be a contract. Then, for any $\sigma\in{LT}$:
	\begin{enumerate} 
		%
		\item Let $\overline{\tau^l_\sigma}$ be the maximal trajectory in $\Pi^{l,in}_\C$ for which there exists a monotonic function from $B^\beta_{{\cal{C}}}(\overline{\tau^l_\sigma})$ into $\sigma$. 
		Then the map $\nu^l_\C:LT\rightarrow\mathbf{B}^{\beta,in}_C$, defined
		by $\nu^l_\C(\sigma)=B^\beta_\C(\overline{\tau^l_\sigma})$ is called the \emph{bundle occurring map} for ${\cal{M}}^l_\C$.
		%
		%
		\item Let $\overline{\tau^e_\sigma}$ be the maximal trajectory in $\Pi^{e,in}_\C$ for which there exists a monotonic function from $L^\rho_{{\cal{C}}}(\overline{\tau^l_\sigma})$ into $\sigma$. 
		Then the map $\nu^e_\C:LT\rightarrow\mathbf{L}^{\rho,in}_C$, defined
		by $\nu^e_\C(\sigma)=L^\rho_\C(\overline{\tau^e_\sigma})$ is called the \emph{transfer occurring map} for ${\cal{M}}^e_\C$.
	\end{enumerate} 
\end{definition}

%
``Maximal'' here means with respect to the ordering defined on trajectories, while the existence of a monotonic function from $B^\beta_{{\cal{C}}}(\overline{\tau^l})$ (from $L^\rho_\C(\overline{\tau^e})$) into $\sigma$ implies that the trajectory $\overline{\tau^l_\sigma}$ ($\overline{\tau^e_\sigma}$) exploits the useful part of $\sigma$.

\begin{theorem}\label{th:monotonic}
	%
	%
	Both $\nu^l_\C{:}LT{\rightarrow}\mathbf{B}^{\beta,in}_C$ and $\nu^e_\C{:}LT{\rightarrow}\mathbf{L}^{\rho,in}_C$
	are monotonic functions.
	%
	%
\end{theorem}

\begin{proof}  
	To prove that $\nu^l_\C$ is a function, we need to prove that $\nu^l_\C(\sigma)$ is uniquely determined. 
	Indeed, consider the first transfer in $\sigma$ completing a bundle, say the bundle $\overline{b}$; then $\overline{b}$ is the first label in the bundle trajectory labeling for the 
	%
	%
	(only) trajectory $\tau^\prime$ consistent with $\sigma$. 
	In the same way, proceeding in $\sigma$ to record the bundles progressively completed by the transfers in $\sigma$ 
	(remember that a bundle can only be completed after reaching a state in which a transition it labels is allowed)
	one obtains the bundle trajectory-labeling of $\tau^\prime$. 
	But now, due to the automaton determinism, $\tau^\prime=\overline{\tau^l_\sigma}$.
	%
	%
	Hence we get a monotonic function from $B^\beta_{{\cal{C}}}(\overline{\tau^l_\sigma})$ into $\sigma$, which selects that part of $\sigma$ which is the linearisation of the bundles in the bundle trajectory-labeling for $\overline{\tau^l_\sigma}$ (the useful subsequence).
	Since ${\cal{M}}_\C$ is deterministic,	$\overline{\tau^l_\sigma}$ is uniquely determined. 
	The proof of the monotonicity of $\nu^l$ is immediate: if $\sigma$ increases, its useful part cannot decrease and it is possible either to reach a further state in the automaton (due to determinism, divergence is impossible) or to remain on the same state. Transfers discarded in the procedure belong to the useless part of $\sigma$.
	The part of the proof relative to $\nu^e_\C$ is trivial, since this map directly relates two sequences of transfers.
\end{proof} 

Being $\nu^l$ and $\nu^e$ monotonic, a smooth correspondence is established between subtrees of $LT$ and 
%
%
of $\mathbf{B}^{\beta,in}_C$, as well as between subtrees of $LT$ and 
of $\mathbf{L}^{\rho,in}_C$.
%
%
%
A monotonic function $\nu^l_e:\mathbf{L}^{\rho,in}_C\rightarrow\mathbf{B}^{\beta,in}_C$ also exists, factorising $\nu^e$ through $\nu^l$.
 
In this model the tree associated with the ledger will also play the role 
of ``making time tick'' with each registration of a new transfer; the present instant at time $t$ is represented by the evolution $\E_t$.

Let us investigate more formally the underlying theory: for a given $\X$ which is an object of $\mathbf{TSL}$, the system $Subtree(\X)$ is seen as a boolean algebra, obtained by equipping the subtrees of $\X$ with natural inclusions between them. 
Hence, we can think of a subtree as of the interpretation of some 
 %
 %
formula.
In this acception, a subtree which is the interpretation of a  
%
%
formula 
$\phi$ is denoted by $\llbracket\phi\rrbracket$.
 Starting from atomic  
 %
 %
 formulas, 
more complex ones are interpreted via the boolean operators present in $Subtree(\X)$.
 
 We thus obtain a notion of satisfiability, 
 %
%
by a path $\pi$, of a logical formula $\phi$:
\begin {itemize}[topsep=2pt,partopsep=1pt,itemsep=2pt,parsep=2pt]
\item $\pi\models_{\X}\phi$ \emph{iff} 
$\pi\in\llbracket\phi\rrbracket$\footnote{Read: $\pi$ satisfies $\phi$ if and only if $\pi$ belongs to the interpretation of $\phi$.}
\item $\pi\models_{\X}\phi\wedge\psi$ \emph{iff} $\pi\in\llbracket\phi\rrbracket\cap\llbracket\psi\rrbracket$
\item $\pi\models_{\X}\phi\vee\psi$ \emph{iff} $\pi\in\llbracket\phi\rrbracket\cup\llbracket\psi\rrbracket$
\item $\pi\models_{\X} \phi\wedge\neg\psi$ \emph{iff} $\pi\notin\llbracket\phi\rrbracket$.
\end {itemize}
 
This logical structure is inherited by every subtree, so that we can relativise the interpretation of a formula to any given subtree, by simply taking the intersection between the extension of the formula and the subtree.

Thus, we have a rigorous tool to vary our satisfiability relation according to the evolution of a ledger, because inclusions between subtrees smoothly translate satisfiability at step $t$ to satisfiability at step $t^\prime<t$.

 With every instantaneous evolution ${\E_t}$ an atomic formula ${\Phi_t}$ is associated. 
 Such 
 %
 %
 a
 formula is taken to mean 
 %
 %
 %
``the ledger state appears in (belongs to) ${\E_t}$''. 
%
 %
For $\E$ an instantaneous evolution (i.e., a tree) and $\phi$ an atomic formula,  
 $\llbracket{\phi}\rrbracket_{{\E}}$ denotes the subtree of ${\E}$ providing the interpretation of $\phi$ in $\E$.
 
 For the set of logical operators above (as well as for the temporal ones later on) one might define a separate notion of satisfiability for every instantaneous evolution 
 %
 %
 $\E_t$. 
 However, since each $\E_t$-interpretation of a formula $\phi$ can be embedded into $\E_0$ by taking the $\E_0$ interpretation of $\phi\wedge{\Phi}_t$ (see Definition~\ref{def:viewSatisfiability}), we prefer to adopt this convention and avoid a useless proliferation of operators.
 %
 
\begin{definition}[Satisfiability in evolutions]\label{def:viewSatisfiability}
	%
	Let 
	%
	%
	$\phi$ be a formula, and let $\sigma$ be a 
	%
	%
	ledger state in ${\E_0}$. 
	We say that $\sigma$ \emph{satisfies} $\phi$ in ${\E_0}$, noted $\sigma\models_{{\E_0}}\phi$,
	%
	%
	according to the following, by induction on the structure of $\phi$:
\begin {itemize}[topsep=2pt,partopsep=1pt,itemsep=2pt,parsep=2pt]
\item $\sigma\models_{{\E_0}}\phi$ \emph{iff} $\sigma\in\llbracket\phi\rrbracket_{{\E_0}}$, 
\item $\sigma\models_{{\E_0}}\phi\wedge\psi$ \emph{iff} $\sigma\models_{{\E_0}}\phi$ and $\sigma\models_{{\E_0}}\psi$;
\item $\sigma\models_{{\E_0}}\phi\vee\psi$ \emph{iff} $\sigma\models_{{\E_0}}\phi$ or $\sigma\models_{{\E_0}}\psi$;
\item $\sigma\models_{{\E_0}}\phi\implies\psi$ \emph{iff} $\sigma\models_{{\E_0}}\phi$ implies $\sigma\models_{{\E_0}}\psi$
\item $\sigma\models_{{\E_0}}\neg\phi$ \emph{iff} it is not the case that $\sigma\models_{{\E_0}}\phi$.
\end {itemize}
We say that $\sigma$ \emph{satisfies} $\phi$ in ${\E_t}$ (i.e., at step $t$), 
noted $\sigma\models_{\E_t}\phi$, 
as follows:
\begin{itemize}[topsep=2pt,partopsep=1pt,itemsep=2pt,parsep=2pt]
\item $\sigma\models_{{\E_t}}\phi$ \emph{iff} $\sigma\models_{{\E_0}}\phi\wedge\Phi_t$.
\end {itemize}
\end{definition}

A similar logical structure can be produced for the meet-semilattices $\mathbf{B}^{\beta,in}_C$ and $ \mathbf{L}^{\rho,in}_C$, so that all the logical operators considered in this section for $LT$, are also definable in $\mathbf{B}^{\beta,in}_\C$ and 
$\mathbf{L}^{\rho,in}_C$.

We now investigate the relationships between these algebraic structures and their associated logics.
We recall that, for a tree $\X$ which is an object in $\mathbf{TSL}$, 
$Subtree(\X)$ is a boolean algebra.  
Theorem~\ref{th:baseChange} unveils a more stringent structure.

\begin{theorem}\label{th:baseChange} 
	Let $\X$ and $\Y$
	be two trees and let $\mu:\X\rightarrow\Y$ be a monotonic function.
	Then, the following hold:
\begin{itemize}[topsep=2pt,partopsep=1pt,itemsep=2pt,parsep=2pt]
\item There exists a monotonic function $\mu^{-1}:Subtree(\Y)\rightarrow{Subtree}(\X)$\footnote{The function $\mu$ is called a \emph{substitution}.}.
	\item The left and the right adjoints to $\mu^{-1}$ exist, namely existential and universal quantifier ($\exists_{\mu}$ and $\forall_{\mu}$).
	\item $\mu^{-1}$ preserves all algebraic operators.
	%
\end{itemize}
\end{theorem}

\begin{proof}(Sketch)
	\begin{itemize}[topsep=2pt,partopsep=1pt,itemsep=2pt,parsep=2pt]
	\item We define $\mu^{-1}$ as the function such that, for $\Y^\prime$ a subtree of $\Y$, $\mu^{-1}(\Y^\prime)=\{\pi\in{X}\mid\mu(\pi)\in{Y}^\prime\}$, which is a subtree of $\X$. 
	Monotonicity is immediate.
	\item The operator $\exists_{\mu}$ is defined, for $\X^\prime$ an object of $\mathbf{TSL}$, by $\exists_{\mu}(\X^\prime)=\{\eta\in{Y}\mid(\exists{\pi}\in{X}^\prime)[\mu(\pi)=\eta]\}$. 
	This is the left adjoint to $\mu^{-1}$, while the operator $\forall_{\mu}$ is defined, for $\X^\prime$ an object of $\mathbf{TSL}$, by
$\forall_{\mu}(\X^\prime)=\{\eta\in{Y}\mid(\forall{\pi}\in{X})[\mu(\pi)=\eta\mbox{ implies }\pi\in{X}^\prime]\}$, which is the right adjoint to $\mu^{-1}$.
%
\item Preservation of operators is a consequence of having left and right adjoints.
\end{itemize}
As all items have been proven, this concludes the proof.
\end{proof}

Hence, we can smoothly translate formulas
in one logic to 
formulas 
in the other one,
the application of $\mu^{-1}$ 
maintaining their syntactical form. 

As we define our logic on a tree-shaped model with a discrete structure,
Definition~\ref{def:operators} introduces the modal/temporal operators of interest, which complement the standard logical connectives.
 
 \begin{definition}[Modal/temporal operators]\label{def:operators}
 Let $\X$ be a prefix-closed tree, let $\leq$ be the prefix relation between its paths and let $\pi$ a path in $\X$. 
 Then: 

\begin {itemize}[topsep=2pt,partopsep=1pt,itemsep=2pt,parsep=2pt]
\item $\pi\models_{\X}\Diamond^{u}\phi$ \emph{iff}
    $(\exists{\pi}^\prime\in\X)[\pi\leq{\pi}^\prime\wedge{\pi}^\prime\models_{\X}\phi]$. 
    In other words: ``there is a future of $\pi$ when $\phi$ becomes true''.

\item $\pi\models_{\X}\Box^{d}\phi$ \emph{iff} $(\forall{\pi^\prime}\in\X)[\pi\leq\pi^\prime\implies{\pi}^\prime\models_{\X}\phi]$.
In other words: ``in all the possible futures of $\pi$, $\phi$ is true''.

\item $\pi\models_{\X}\Diamond^{d}\phi$ \emph{iff}
    $(\exists{\pi}^\prime\in\X)[\pi^\prime\leq\pi\wedge{\pi}^\prime\models_{\X}\phi]$.
    In other words: ``there is a past of $\pi$ when $\phi$ was true''.

\item  $\pi\models_{\X}\Box^{u}\phi$ \emph{iff}
    $(\forall{\pi}^\prime\in\X)[\pi\leq\pi^\prime\implies{\pi}^\prime\models_{\X}\phi]$.
    In other words: ``for all pasts of $\pi$ $\phi$ was true''\footnote{Note that ``future'' and ``past'' correspond to a ``d''-index (for down) and to a ``u''-index (for up), respectively, for the $\Box $ operators, while it is the other way around for the $\Diamond$ operators.
    }.
\end{itemize}
 
 Using the strong partial order $<$ canonically associated with the weak partial order $\leq$, 
 we define the conditions for the relation $Succ(\pi^\prime,\pi)$ to hold as: 
 $\pi<\pi^\prime$ and there does not exist $\pi^{\prime\prime}$ such that
 $\pi < \pi^{\prime\prime} < \pi^\prime$. 
 The following  ensues:
  
  \begin{itemize}[topsep=2pt,partopsep=1pt,itemsep=2pt,parsep=2pt]
      \item $\pi\models_{\X}\mbox{NXT}_{\Diamond}\phi$ \emph{iff}
          $(\exists{\pi^\prime})[Succ(\pi^\prime,\pi)\wedge{\pi^\prime}\models_{\X}\phi]$, i.e.,
          ``there is an immediate future of $\pi$ when $\phi$ becomes true''.
    \item $\pi\models_{\X}\mbox{NXT}_{\Box}\phi$ \emph{iff} 
    $(\forall{\pi}^\prime)[Succ(\pi^\prime,\pi)\implies{\pi^\prime}\models_{\X}\phi]$,  i.e.,
         ``for all immediate futures of $\pi$, $\phi$ becomes true''.
 \end {itemize}
 
 Analogously, if we define that the relation $Pred(\pi^\prime,\pi)$ holds if $Succ(\pi,\pi^\prime)$ holds, we have the following:
 
 \begin{itemize}[topsep=2pt,partopsep=1pt,itemsep=2pt,parsep=2pt]
      \item $\pi\models_{\X}\mbox{PRV}_{\Diamond}\phi$ \emph{iff}
          $(\exists{\pi^\prime})[Pred(\pi^\prime,\pi)\wedge{\pi^\prime}\models_{\X}\phi]$, i.e.,
         ``there is an immediate past of $\pi$ when $\phi$ was true''.
    \item $\pi\models_{\X}\mbox{PRV}_{\Box}\phi$ \emph{iff} 
    $(\forall{\pi}^\prime)[Pred(\pi^\prime,\pi)\implies{\pi^\prime}\models_{\X}\phi]$,  i.e., ``for all immediate pasts of $\pi$ $\phi$ was true''
 \end {itemize}
  \end {definition}
  
  \begin{theorem}\label{th:temporalOperators}
  The temporal operators in Definition~\ref{def:operators} enjoy the following algebraic properties:
  
   \begin{enumerate}[topsep=2pt,partopsep=1pt,itemsep=2pt,parsep=2pt]
   \item They are all monotonic functions.
   \item Squares and diamonds form temporal doctrines (see~\cite{BGKL18}). In other words, they are left (right) adjoint (and also left (right) inverse) to inclusions of up or down completion: namely
   \begin{itemize}
   \item $\Diamond^{u}\phi$ corresponds to the minimal subobject containing $\llbracket{\phi}\rrbracket$ complete w.r.t. prefixes.
    \item $\Box^{u}\phi$ corresponds to the maximal subobject contained in $\llbracket{\phi}\rrbracket$ complete w.r.t. prefixes.
     \item $\Diamond^{d}\phi$ corresponds to the minimal subobject containing $\llbracket{\phi}\rrbracket$ complete w.r.t. prolongations.
      \item $\Box^{d}\phi$ corresponds to the maximal subobject contained in $\llbracket{\phi}\rrbracket$ complete w.r.t. prolongations.
     \end {itemize}
   \item Next and previous operators also form temporal doctrines: they correspond to left (right) adjoint to inclusions of next step or previous step completion.
   \end {enumerate}
  \end {theorem}
  \begin{proof}  The results can be obtained routinely in analogy with those in~\cite{BGKL18}.
    \end{proof}

     \begin{corollary}
     All of the properties of the operators in Definition~\ref{def:operators} (in particular their interrelations) 
     are completely defined by their being associated with adjoint functions, hence uniquely determined.
     \end{corollary}

 Now we are ready to show how to define a property in one tree also using the other one. 
 To this end, we look at the interpretation in $LT$ of the atomic formula corresponding to a property w.r.t. a given ${\cal{M}}^l_{\cal{C}}$.

\begin{itemize}[topsep=2pt,partopsep=1pt,itemsep=2pt,parsep=2pt]
\item $\sigma \in LT$ is \emph{bundle-complete} (noted $\sigma\models{bc}$) \emph{iff} $\sigma^\prime<\sigma\implies\neg(\nu^l(\sigma^\prime)=\nu^l(\sigma))$.
\item $\sigma$ is \emph{contract-safe} if it satisfies $\Diamond^{u}{bc}$. 
\end{itemize}

The first condition means that $\sigma$ reaches a stable state of the automaton, 
while the second one means that $\sigma$ does not contain any spurious transfer.

\section{The associated language at work}\label{sec:language}
The definition of the logical (modal/temporal) operators in Section~\ref{sec:categoricalModel} immediately lends itself to the definition of a query language, where 
a query is formulated in terms of paths in a suitable tree, satisfying a given formula.
To start with, the 
language contains two atomic sentences, presented in Definition~\ref{def:atomicSentences}, to be interpreted in the structure $LT$.  
%
%
%
%

\begin{definition}[Atomic sentences]\label{def:atomicSentences}
Let $LT$ be the tree corresponding to the 
%
%
evolution $\E_0$ for a given ledger built on an alphabet of transfers $T$. 
Then the following atomic sentences are defined by producing the respective interpretations.
\begin{itemize}[topsep=2pt,partopsep=1pt,itemsep=2pt,parsep=2pt]
\item For any ledger state $\sigma$, we consider the property $\leq_t$, which is satisfied if $|\sigma|\le{t}$.
In $LT$ this corresponds to identifying the subtree $\chi_t=\llbracket\leq_t\rrbracket_{LT}$ composed of all the 
sequences in $LT$ which are not longer than $t$.

\item For a transfer $\theta\in T$, the operator $app_{\theta}$ is interpreted in $LT$ as the set of sequences 
in which $\theta$ appears\footnote{Actually, the paths in $LT$ contain encodings of transfers, but for the sake of simplicity in the rest of the section we identify the two notions.}, 
i.e., it is the subtree $\llbracket{app}_{\theta}\rrbracket_{LT}$. 
\item For a transfer $\theta\in T$, 
the operator $app_{\theta,n}$ is interpreted in $LT$ as the set of sequences in which $\theta$ appears in position $n$, 
i.e., as the subtree $\llbracket{app}_{\theta,n}\rrbracket_{LT}$. 
\end{itemize} 
\end{definition}

The following facts test the expressivity of the language built with the atomic sentences of Definition~\ref{def:atomicSentences} and the operators of Definition~\ref{def:operators}.

\begin{fact}\label{th:subobjects}
Given a set of ledger states $LT$ for a given ledger and a contract $\cal{C}$, we can define the following subtrees as interpretations of the respective formulas (the first one refers to the ledger, the other ones to a contract):

\begin{enumerate}[topsep=2pt,partopsep=1pt,itemsep=2pt,parsep=2pt]
\item\label{it:sub1} The interpretation in ${LT}$ of $appLst_{\theta,t}=\bigvee_{1\leq{n}\leq{t}}(app_{\theta,n}\wedge\chi_n)$ is given by the tree of ledger states $\sigma$, of length at most $t$, with $\theta$ as last transfer in $\sigma$. 

\item\label{it:sub2} Given a bundle initial trajectory-labeling $z_B$ in ${\mathbf{B}}^{\beta,in}_\C$,
let $\zeta_B$ be the formula interpreted into the corresponding singleton.
Then the set of bundle initial trajectory-labelings in $\mathbf{B}^{\beta,in}_\C$
such that they start with $z_B$, i.e. the set of prolongations of $z_B$,
is the interpretation in $\mathbf{B}^{\beta,in}_\C$ of 
%
%
$\Diamond^d{\zeta}_B$.

\item\label{it:sub3} Given a transfer initial trajectory-labeling $z_L$ in ${\mathbf{L}}^{\rho,in}_\C$, 
let $\zeta_L$ be the formula interpreted into the corresponding singleton.
Then the set of transfer initial trajectory-labelings in ${\mathbf{L}}^{\rho,in}_\C$
such that they start with $z_L$, i.e. the set of prolongations of $z_L$, is the interpretation in ${\mathbf{L}}^{\rho,in}_\C$ of 
%
%
$\Diamond^d{\zeta}_L$.

\item\label{it:sub4} Given the translation function ${\nu^l_{\C}}^{-1}$ (resp. ${\nu^e_{\C}}^{-1}$) and $z_B$ a labeling in $\mathbf{B}^{\beta,in}_\C$ 
(resp. $z_L$ a labeling  in $\mathbf{L}^{\rho,in}_\C$), the set of ledger states
where a prefix of $z_B$ (resp. $z_L$) has been performed reaching a ``legal'' state in ${\cal{M}}^l_{\C}$ 
(resp. a possibly ``intermediate'' state in ${\cal{M}}^e_{\C}$), is the interpretation in $\E_0$ of 
${\nu^l_{\C}}^{-1}(\Diamond^d(\zeta_B))$ (resp. of ${\nu^e_{\C}}^{-1}(\Diamond^d (\zeta_L))$).

\item\label{it:sub5} The set of ledger states in $LT$, noted $hol_{(r,k,t)}$, such that, as a result of the transfers in this sequence, an agent $k$ holds a resource $r$ at time $t$ (w.r.t. a contract $\C$), is the interpretation of the formula 
$\Diamond^d(\bigvee_{k^\prime}\bigvee_{t^\prime\leq{t}}appLst_{(r,k^\prime,k),t^\prime} \wedge\Box^d(\neg\bigvee_{k^{\prime\prime}}\bigvee_{t^\prime\leq{t}^{\prime\prime}\le{t}}appLst_{(r,k,k^{\prime\prime}),t^{\prime\prime}}))$. (That is to say: at time $t$, the token $r$ is with the agent $k$, since at some previous time $t^\prime$ it was transfered to $k$ from $k^\prime$, and in no prolongation of the state $t^\prime$, namely at $t^{\prime\prime}$ which is a prefix of the state at $t$, a transfer towards some actor $k^{\prime\prime}$ has been recorded.)
\end{enumerate} 
\end{fact}

We can now express in the language (and, possibly, use as axioms in deductions) the properties required in Section~\ref{sec:contracts}, and assume that the following formulas are always verified when interpreted in $\mathbf{B}^{\beta,in}_\C$:

\begin{itemize}
\item $\bigwedge_r\bigwedge_k\neg{app}_{(r,k,k)}$, meaning: ``no transfer of the form $(r,k,k)$ occurs''; 
\item $\bigwedge_r\bigwedge_k\neg{app}_{(r,\bot,k)}$ meaning: ``no transfer of the form $(r,\bot,k)$ occurs'';
\item $\bigwedge_r\bigwedge_k (app_{(r,k,\bot)}\implies\Box^d\bigwedge_{k^\prime}\bigwedge_{k^{\prime\prime}}\neg{app}_{(r,k^\prime,k^{\prime\prime})})$, meaning:
``once $(r,k,\bot)$ has occurred, then no transfer for $r$ can appear in any future''.
\end{itemize}
Similarly, we can assume that the formula $\Phi_{t+1}\implies\Phi_{t}$
is always verified when interpreted in $\E_0$, thus formalising as an axiom of the logic 
the fundamental property of \emph{immutability} of ledgers. 
Indeed, the formula expresses 
that each evolution at time $t+1$ is contained in the evolution at time $t$, thus formalising the fact that
``once a ledger state is reached, it is reached forever''.

\subsection{Formalising queries}
We now give some examples of queries we can ask a ledger w.r.t. a contract. Let us assume that, when we write $\sigma_i$ for a ledger state, we intend $|\sigma_i|=i$, i.e. $\sigma_i$ satisfies $\chi_i \wedge \neg \bigvee_{0\leq k \leq i-1} \chi_k$. Hence $\sigma_i$ satisfies $\phi$ will mean that this fact happens at the instant $i$. 
Recall that an evolution $\E_j$ has only one state $\sigma_i$ for $i\leq j$ (a past, established state) and only prolongations of $\sigma_j$ as future states.

\begin{enumerate}
	\item  Suppose we start recording on a ledger the transfers relative to contract $\C_d$ of Example~\ref{ex:contract} performed according its associated contract execution automaton ${\cal{M}}^e_\C$ (with the useful part reaching a state in ${\cal{M}}^l_\C$).
	The encodings of transfers can be registered on the ledger at different (not necessarily consecutive) instants of time.
	We can then ask ``which is the state of the legal contract (resp. contract execution) automaton for a state of the ledger $\sigma_8$, in an evolution $\E_n$, with $8\le{n}$?''  
	The answer is given by evaluating the formulas
	$\exists_{\nu^e}(\Phi_n\wedge\chi_8)$, 
	with respect to	the contract execution automaton, and  
	$\exists_{\nu^l}(\Phi_n\wedge\chi_8)$, with respect to the contract legal automaton.
	
	\item A simple historical query made possible by our approach is the following: 
	given an interval $[i\dots{j}]$ of time instants, and fixed an evolution $\E_s$, $j\le{s}$, 
	we can ask whether a formula $\phi$ is true for each ledger state $\sigma_k$ 
	with $k\in[i\dots{j}]$ in the evolution $\E_s$.
	A typical example for $\phi$ is:
	$hol_{(\mbox{\emph{Eiffel}},a,all(X[i,j]))}$, meaning that 
	``In the interval between $i$ and $j$, Alice (represented as $a$) owns the non-fungible token certifying the property of Eiffel tower''.
	Note that this notation seems to use a new universal quantification, but it is actually syntactic sugar: the query can be rolled out into a finite conjunction of queries the form 
	$hol_{(\mbox{\emph{Eiffel}},a,i)}\wedge\dots\wedge{hol}_{(\mbox{\emph{Eiffel}},a,j)}$, as we have made in other cases. 
	Hence, the query is answered by the satisfiability of $\phi\wedge\Phi_s$.
	\item 
	Many variations on this theme are possible. For example: given $\sigma_i$, 
 we could verify whether, in the future of it there will be a state $\sigma_j$ verifying $\phi$, provided that, in its past the series of transfers $\theta_{1},..., \theta_{k}$, at times $j_1\leq...\leq j_k$,  has been recorded.
 This is equivalent to verifying the formula: 
 $\Diamond^u(\Diamond^d(\bigwedge_{1\leq{n}\leq{k}}app_{(\theta_n,j_n)})\implies\phi)$.
	
	\item A further variation is a form of hypothetical reasoning: given an evolution $\E_i$ and a formula $\phi$ which is not true in $\sigma_i \in \E_i$, is there a prefix $\sigma_k$ of $\sigma_i$ 
	 such that $\phi$ could be true in some future of $\sigma_k$?
	In this case, we are dealing with a ledger state $\sigma_i$ in the interpretation of $\neg\phi\wedge\Diamond^d\Diamond^u\phi \wedge \Phi_i$.
	
	\item  A simple, but interesting, type of query is about the possibility of verifying, at given instant $i$, the presence of sequences of transfers registered on the ledger. 
	Without loss of generality, we consider sequences composed of two transfers: $\langle\theta_1\theta_2\rangle$.
	Then the occurrence of such a sequence on a  state 
	%
	%
	is expressible by the formula:
	$\Diamond^d(app_{(\theta_2,n)}\wedge \Diamond^d(app_{(\theta_1,n^\prime)}))$. 
	Queries of this kind can be structured in more complex ones like
	$
	\Diamond^d(app_{(\theta_2,n)}\wedge\Diamond^d(app_{(\theta_1,n^\prime)}))\wedge
	\Diamond^u(\Diamond^d(app_{(\theta_2,m)})\wedge\Diamond^d(app_{(\theta_1,m^\prime)})) 
	$, 
	looking for repetitions of the same sequence later in the time.
	One could weaken the condition that the transfers occurring in the two sequences are ordinately equal, by imposing only a certain kind of similarity: e.g., the same actors, or the same resource, etc.
\end{enumerate}

Some remarks are in order.

Technically, the construction of the interpretation of the formula in Item 3 is equivalent to a ``temporal projection via regression'': transfers are not executed, but consequences of their hypothetical execution are verified.
In practice, temporal regression queries must be handled adequately to ensure that meaningful answers return and unnecessary computational burdens are avoided. 
In fact, in themselves they are subject to infinite answers: think of the case in which, starting from a state of affairs in which Bob is in possession of a token that certifies his ownership of the Colosseum and Alice is in possession of a token that certifies her ownership of the Eiffel tower, we want to verify the possibility of switching ownerships, with Bob and Alice taking possession, respectively, of the French Belle Epoque building and of the Roman era stadium. 
One way for this to happen is by exchanging tokens, if contracts have been defined that enable such exchanges; but this may happen an arbitrary number $n$ of times, by sending them back and forth for cases greater than 1, while one is presumably interested solely in the case $n = 1$. 
In a concrete ledger management system, managing queries of this type
%
%
could be delegated to the user, by allowing her to place limits on the time $j_k$ when the hypothesized situation occurs in the future, or to the system, by constraining it to return only the shortest answer.
	
Queries of the type discussed in Item 5 are particularly useful and interesting because they can act as tools for auditing and process mining on the activities of the ledger. For example, they could be used to verify the repetition over time of the sending of quantities of money from Bob to Alice and then from Alice to George, which could be indicative of a laundering activity. Or they could identify the repetition, for instance, of supplies of kitchen plinths from company A to company B, which in turn uses them to assemble kitchen furniture which then routinely supplies to company C that produces turnkey kitchens, fully furnished and equipped with appliances; these repetitions may suggest connecting the three companies in a single supply chain through optimized revenue sharing  contracts as illustrated in \cite{BGMP+20,BGMP+21}.


\section {Discussion and related work}\label{sec:related}
Not long after its inception, database technology produced several temporal data models~\cite{Sno99,Sno95,ACGJ+93}. 
However, these models apply to data, like bank accounts and document versions, produced in application contexts where time is treated as an add-on to basic data models bereft of a time dimension. 
By contrast, in ledgers
time is not an option, but stands out as a fundamental, albeit implicit, aspect in the overall functioning of the system. 
Hence, our approach differs substantially from models such as those mentioned above precisely because we treat time as a universal factor independent of specific data types, starting from the assumption that all data in a ledger is set in a time dimension.

Much more recently,the works in~\cite{BRDD18} and~\cite{YCG19} describe systems that use relational database technologies to extract information from the Ethereum and Bitcoin blockchains by leveraging the available meta-information, such as hash numbers of blocks and transactions.
Thus, while having significant practical usability, they are very coarse-grained and very specific to the environments they are meant for, hence they do not provide for generalizable data models with a built-in management of time.

Quite closer to our approach is the treatment of time in the Situation Calculus~\cite{CH81,Rei01}, a formal framework that provides a general theory of action by viewing the world as a succession of situations in consequence of the execution of admissible actions. 
It has in this sense influenced our elaboration of similar formal and computational tools to support querying in ledgers.
On the other hand, the primary applications of Situation Calculus are in the areas of knowledge representation for planning by artificial agents, thus requiring a rich vocabulary made up of fluent predicates describing various aspects of the world relevant for the agents’ actions.  This results into “frame problems” and “ramifications problems” in the update of such representations that need to be tackled whenever even seemingly minimal changes in the state of affairs take place. By contrast, our lean representation structure links actors and resources through a single graph, thus staying at large from laborious update processes.  
At a logical level, the Situation Calculus can be reconstructed through modal/temporal logics~\cite{Ben11,Lak10} bearing a number of resemblances with, and sharing characteristics of, the logics used here to provide the basis to query past, present and future state of affairs of the 
%
%
ledger.
In particular, the modal/temporal logic versions of the Situation Calculus share with the approach we have presented here the management of state of affairs as “possible worlds” rather than as reified entities referred to through second-order variables as in the original version of the Situation Calculus, with considerable simplifications in the technical treatment of the underlying semantics. 
Temporal and dynamic logics are used in~\cite{Hal17,MGOS19,BDS20} to address dynamic aspects of the blockchain, orthogonal to those discussed here, such as, respectively, the guarantees given by the validation protocols to agents involved in contractual transactions and the management of provisional blocks.

The relationships between legal contracts, smart contracts, transactions and blockchains are wide and varied and unfold from a multiplicity of applications and of disciplinary contexts. 
Long before blockchains became commonly used technologies, and quite independently of them, the concept of smart contract was introduced by Nick Szabo in the 1990s~\cite{Sza97}. 
The focus and fundamental motivation for this conceptualization stems from the observation made by Szabo that the contracts of the legal tradition, especially those pertinent to the commercial sphere, establish that transactions, in the commercial sense of the term, must be carried out in conjunction with given events, such as the payment deadline for a rent, mortgage or leasing contract. 
In the treatment given by Szabo, a contract can therefore be seen as the automation of the transactional part of a legal contract, with the guarantee that the transactions will be effectively carried out, as well as the ability to automatically manage infringements of the agreed conditions without having to resort straight away to legal proceedings - for example by automatically sending a warning letter in the case of an overdue payment, or by electronically locking access to rented premises in case of accumulation of non-payments. With this conceptualization of smart contracts, Szabo aimed to increase the efficiency in the execution of the underlying legal agreements, as well as to decrease legal disputes arising from any infringements. 
However, time was not yet ripe for the concept to pass into practice and implementation, and it will take more than two decades for Szabo's pioneering contribution to gain light and consideration.

In the meantime, the idea was born of transferring the technology of advanced DBMS transactional models beyond the fences of the single company, so as to leverage them to support business eco-systems made possible by the concurrent advent of technological trends such as Web services, service-oriented architectures and, generally, of distributed systems hinging on cross-entity computer interactions. 
Early attempts to adapt the most advanced of such models, based on the properties of Atomicity, Consistency, Isolation and Durability (ACID), to implement complex e-commerce transactions spanning multiple businesses are described in~\cite{AFP96,APP97,APPP98}. 
However, substantial difficulties faced by such models were rooted in the fact that the short-lived nature of ACID transactions as conceived for single corporate domains did not fit well with the longer-life requirements of multi-business transactions. 
Alternatives to the timed-out resource-locking of traditional ACID models were thus introduced, in particular compensating transactions~\cite{KS03}, although these too proved insufficient to make the envisaged business model fly, in the absence of a trustable infrastructure for cross-enterprise transaction execution. 
Meanwhile, a transaction model such as BASE 
(Basic Availability, Soft-state, Eventual consistency)\footnote{www.dataversity.net/acid-vs-base-the-shifting-ph-of-database-transaction-processing}, characterized by weaker constraints with respect to ACID, was successful in the practice of distributed applications such as social networks and e-commerce stores, that can be handled without taking care of demanding contractual requirements.

But the coming of age of blockchains showed that the last word had yet to be said about the feasibility of bringing together the world of commercial contracts and of digital transactions. Indeed, second-generation blockchains like Ethereum revived and brought all the previous attempts together, by leveraging the digital trust afforded by blockchain technology to resurrect and implement Szabo's proposal.  
Yet, here too there are a number of weaknesses to address, as smart contracts implemented in Ethereum and other 
blockchains~\cite{GDAD+20} bring together basic transactions but are not inherently transactional, in the sense of being bound by a well-defined and verifiable set of properties that define the effects of their execution on the database, while they are instead treated as programming constructs in dedicated or conventional programming languages. 
Specialised languages and frameworks for smart contracts include Solidity on Ethereum\footnote{https://docs.soliditylang.org/en/latest/\#,https://ethereum.org/en/} and private distributed ledgers maintaining records of all the performed (atomic) transactions like Fabric from the Hyperledger software ecosystem \footnote{https://www.hyperledger.org/use/fabric, https://www.hyperledger.org}. 
As these languages and frameworks usually lack a proper formal foundation, but may, on the other hand, be Turing-complete~\cite{ZXDC+20}, problems arise with respect to verification, both of their behavioural properties and of the conformity between a contract's textual definition and its programmatical implementation. 
This has led to a number of problems like the already mentioned notorious concurrency bug that affected the DAO smart contract through which Ethereum itself was funded~\cite{DP17}. 
By contrast, our formalization of a contract as an automaton forces a rigorous division of its various execution phases, so as to provide a declarative model transparent to aspects of access to resources. 
 Furthermore, the ability, intrinsic to our formalization, to impose interdependencies between the various transactions that make up a smart contract substantially reproduces the atomic characteristics of the ACID paradigm, without however binding them to the release and roll-back of resources after time-out, which would disagree with contexts of use populated by actors that can be flexible on their times to respond to requests. 

At the same time, our model is sufficiently general to be independent from any implementation infrastructure 
(also the original formulation by Szabo abstracted from specific execution models), 
and yet compatible with various deployments on ledgers, distributed ledgers and blockchains. 
Finally, it could also be synthesized from optimization procedures as in the case of the ``intelligent smart contracts'' for innovative supply chain management illustrated in~\cite{BGMP+20,BGMP+21} with implementation, respectively, in  a public blockchain (Ethereum) and a private distributed ledger (Hyperledger Fabric).

Another aspect for which the implementation of smart contracts as programs is unsatisfactory is that the relationship with legal contracts is substantially weakened, if not lost altogether, a relationship whose reconstruction is a specific objective of this work. 
In this sense, there are parallels and convergence with other efforts in a similar direction. 
The work in~\cite{MMN17} first recognized that to make a legal contract computer-executable requires the availability of suitable formal models for specifying its intended (legal) semantics. 
As concerns this latter point, two proposals, both named Contract Automata and based on similar principles, are contained in~\cite{APSS16} and~\cite{BDF16}.

In~\cite{APSS16}, states of the automaton correspond to sets of \emph{permissions} and \emph{obligations} to execute some actions. 
A state transition occurs when a permitted action in a state is performed, giving rise to a new configuration of permissions/obligations.
The main goal is to verify the soundness of the contract, for example to check that any action for which there is an obligation in a given state is also permitted in that state. 
The treatment is restricted to two-party contracts, so that a state is annotated with a four-tuple of sets of deontic clauses, describing, for each party, what its permitted actions and its obligations are.
No explicit consideration is given to the resources needed to perform the action, as the assumption is that if a part $p$ has permission to perform an action $a$, $p$ will succeed in performing $a$, if it tries.
The authors suggest that the possibility of checking the actual feasibility of $a$ could be modeled by introducing a different action $attempt\_a$ which then becomes the object of the permission.
This extension could open the way for an explicit discussion of the conditions, i.e., of the resources needed, for an attempt to be successful.

The model in~\cite{BDF16} considers multi-party contracts, where each participant is modeled by an automaton, and a contract describes the conditions under which transitions in the different automata can be synchronised.
Here, each party \emph{offers} or \emph{requests} the execution of some action, and synchronisation is achieved when requests by one party are matched by offers from another one. 
In this case, offers and requests may concern resources produced by a participant and consumed by another one.
Synchronisation can also be weakened, provided that requests and offers are matched in the end, thus introducing a form of credit, to be honoured before the end of the contract. 
Various logics for reasoning on contracts are then derived and compared.

Finite-state machines are also at the basis of VeriSolid~\cite{ML18,MLSD19}, a tool to specify contracts to be deployed on Ethereum. 
The generated contracts are secure by design, as they do not allow some sequences of actions (invocations of a Solidity function), thus reducing the expressiveness of Solidity. 
Plugins can impose further constraints,
%
%
for instance requiring that functions are invoked according to the linear order provided by some shared program counter.

These models do not consider resources in the definition of state, and, lacking a notion of transaction as set of actions, transitions occur on single actions. 

A connection between the notion of contract (actually, the almost equivalent notion of \emph{policy}) and the notion of resource has been drawn in~\cite{BFHP17}. 
There, policies are defined in terms of admissible paths in the space state of some actor subject to the policy. 
State transitions have to be synchronised with corresponding transitions in the state of some resource necessary to realise the transition (where the resource state typically describes the availability of the resource for some actor). 
This is achieved by annotating states in a policy with resources, and expressing synchronisation through constraints on valid annotations. 
The model of allocations of resources to actors adopted in this paper can be seen as equivalent to the form of annotation adopted there, especially since allocation systems fully support the notion of state for both resources and actors.

\section{Conclusions}\label{sec:concls}
The approach to reasoning on evolutions of ledgers introduced in this paper starts from 
a formalisation of actions required by a contract in terms of transfers of tokens representing specific resources between specific actors; 
%
%
the recording of these transfers on ledgers 
%
marks the advancements in the contract's execution.
This provides the basis for the construction of computational solutions that appear advantageous and appropriate for the evolution of the technological paradigm to which these constructs are associated.  
These solutions have been expressed at a level of maximum generality and abstraction so as not to place restrictions on their transferability to a wide variety of concrete implementation contexts.  
Specifically, we have shown that the resulting notion of ledger state is compatible with query mechanisms that can be translated into a modal/temporal logic where it is possible to query the present, the future and the past of the ledger, 
as well as its ``present perfect'' and ``nearest future'', and also to reason counterfactually with respect to what would have happened if the course of things, that is, of transactions, had been different.  
Maximum flexibility is thus guaranteed in the auditing of activities on the ledger.  
%
%
Moreover, the formal characterization of contracts as automata
%
%
allow us to 
treat them as advanced transactional models,  rather than as programs, as is the case in the current practice of smart contract implementations.  
This model reproduces some of the properties of the  well-established  ACID  model  in  the context of ledgers and contracts, such as consistency, atomicity in the sense of interdependence between transactions, declarativeness in the access to resources used for the completion of contractual obligations.  
It therefore generalizes, at an abstract and formal level, indications and proposals for an effectively transactional treatment of smart contracts, which is desirable in order to avoid programming errors that have plagued their practice in the past, sometimes with disastrous financial consequence, and to boost effectiveness in contract execution.  
The two contributions come together, as the query logic can be used to audit the progress and actual execution of contracts built on the basis of these formal criteria.

\section*{Acknowledgements}
Work partially supported by Sapienza, project ``Consistency problems in distributed and concurrent systems''. 
We thank the anonymous referees for indication on how to improve this paper.

\bibliography{biblio}

\begin{thebibliography}{10}
\expandafter\ifx\csname url\endcsname\relax
  \def\url#1{\texttt{#1}}\fi
\expandafter\ifx\csname urlprefix\endcsname\relax\def\urlprefix{URL }\fi
\expandafter\ifx\csname href\endcsname\relax
  \def\href#1#2{#2} \def\path#1{#1}\fi

\bibitem{Sza97}
N.~Szabo, Formalizing and securing relationships on public networks, First
  Monday 2~(9).

\bibitem{But13}
V.~Buterin, \href{https://ethereum.org/en/whitepaper/}{Ethereum whitepaper}
  (2013).
\newline\urlprefix\url{https://ethereum.org/en/whitepaper/}

\bibitem{BGKL18}
P.~Bottoni, D.~Gorla, S.~Kasangian, A.~Labella, A doctrinal approach to
  modal/temporal heyting logic and non-determinism in processes, Mathematical
  Structures in Computer Science 28~(4) (2018) 508--532.

\bibitem{BGKL19}
P.~Bottoni, D.~Gorla, S.~Kasangian, A.~Labella, Modal epistemic logic on
  contracts: {A} doctrinal approach, in: Models, Languages, and Tools for
  Concurrent and Distributed Programming - Essays Dedicated to Rocco De Nicola
  on the Occasion of His 65th Birthday, Vol. 11665 of LNCS, Springer, 2019, pp.
  298--314.

\bibitem{AL21}
P.~Bottoni, A.~Labella, Transactions and contracts based on reaction systems,
  Theoretical Computer Science (2021) 1--50\href
  {http://dx.doi.org/10.1016/j.tcs.2021.07.012}
  {\path{doi:10.1016/j.tcs.2021.07.012}}.

\bibitem{KL99}
S.~Kasangian, A.~Labella, Observational trees as models for concurrency,
  Mathematical Structures in Computer Science 9~(6) (1999) 687--718.

\bibitem{BGMP+20}
P.~Bottoni, N.~Gessa, G.~Massa, R.~Pareschi, H.~Selim, E.~Arcuri, Intelligent
  smart contracts for innovative supply chain management, Frontiers in
  Blockchain 3 (2020) 52.

\bibitem{BGMP+21}
P.~Bottoni, N.~Gessa, G.~Massa, R.~Pareschi, D.~Tortola, Distributed ledgers to
  support revenue-sharing business consortia: a hyperledger-based
  implementation, in: Proc. {BRAIN} 2021, 2021, to appear.

\bibitem{Sno99}
R.~T. Snodgrass, Developing Time-Oriented Database Applications in SQL, Morgan
  Kaufmann, 1999.

\bibitem{Sno95}
R.~T. Snodgrass, The TSQL2 Temporal Query Language, Kluwer, 1995.

\bibitem{ACGJ+93}
A.~U. Tansel, J.~Clifford, S.~K. Gadia, S.~Jajodia, A.~Segev, R.~T. Snodgrass,
  Temporal Databases: Theory, Design, and Implementation, Benjamin/Cummings,
  1993.

\bibitem{BRDD18}
S.~Bragagnolo, H.~Rocha, M.~Denker, S.~Ducasse, {Ethereum Query Language}, in:
  Proc. WETSEB@ICSE 2018, {ACM}, 2018, pp. 1--8.

\bibitem{YCG19}
K.-B. Yue, K.~Chandrasekar, H.~Gullapalli, Storing and querying bitcoin
  blockchain using {SQL} databases 17 (2019) 24--41.

\bibitem{CH81}
J.~McCarthy, P.~Hayes, Some philosophical problems from the standpoint of
  artificial intelligence, in: B.~L. Webber, N.~J. Nilsson (Eds.), Readings in
  Artificial Intelligence, Morgan Kaufmann, 1981, pp. 431--450.

\bibitem{Rei01}
R.~Reiter, Knowledge in Action: Logical Foundations for Specifying and
  Implementing Dynamical Systems, The MIT Press, 2001.

\bibitem{Ben11}
J.~{van Benthem}, {McCarthy} variations in a modal key, {Artificial
  Intelligence} 175~(1) (2011) 428--439.

\bibitem{Lak10}
G.~Lakemeyer, The situation calculus: A case for modal logic, Journal of Logic,
  Language, and Information 19~(4) (2010) 431--450.

\bibitem{Hal17}
J.~Y. Halpern, R.~Pass, A knowledge-based analysis of the blockchain protocol,
  Electronic Proceedings in Theoretical Computer Science 251 (2017) 324–335.

\bibitem{MGOS19}
B.~Marinkovic, P.~Glavan, Z.~Ognjanovic, T.~Studer, A temporal epistemic logic
  with a non-rigid set of agents for analyzing the blockchain protocol, J. Log.
  Comput. 29~(5) (2019) 803--830.

\bibitem{BDS20}
K.~Brünnler, D.~Flumini, T.~Studer, {A logic of blockchain updates}, Journal
  of Logic and Computation 30~(8) (2020) 1469--1485.

\bibitem{AFP96}
J.~Andreoli, S.~Freeman, R.~Pareschi, The coordination language facility:
  Coordination of distributed objects, Theory Pract. Object Syst. 2~(2) (1996)
  77--94.

\bibitem{APP97}
J.~Andreoli, F.~Pacull, R.~Pareschi, {XPECT:} {A} framework for electronic
  commerce, {IEEE} Internet Comput. 1~(4) (1997) 40--48.

\bibitem{APPP98}
J.~Andreoli, F.~Pacull, D.~Pagani, R.~Pareschi, Multiparty negotiation of
  dynamic distributed object services, Sci. Comput. Program. 31~(2-3) (1998)
  179--203.

\bibitem{KS03}
R.~Karlsen, T.~Stranden{\ae}s, Trigger-based compensation in web service
  environments, in: Proc. {ICEIS}, 2003, pp. 487--490.

\bibitem{GDAD+20}
M.~Garriga, S.~Dalla~Palma, M.~Arias, A.~De~Renzis, R.~Pareschi, D.~A.
  Tamburri, Blockchain and cryptocurrencies: A classification and comparison of
  architecture drivers, Concurrency and Computation: Practice and Experience
  e5992 (2020) 1--21.

\bibitem{ZXDC+20}
Z.~Zheng, S.~Xie, H.-N. Dai, W.~Chen, X.~Chen, J.~Weng, M.~Imran, An overview
  on smart contracts: Challenges, advances and platforms, Future Generation
  Computer Systems 105 (2020) 475--491.

\bibitem{DP17}
Q.~DuPont, Experiments in algorithmic governance: {A} history and ethnography
  of ``{The DAO}", a failed {Decentralized Autonomous Organization}, in:
  M.~Campbell-Verduyn (Ed.), Bitcoin and Beyond, Routledge, 2017, pp. 157--177.

\bibitem{MMN17}
D.~Magazzeni, P.~McBurney, W.~Nash, Validation and verification of smart
  contracts: A research agenda, {IEEE} Computer 50~(9) (2017) 50--57.

\bibitem{APSS16}
S.~Azzopardi, G.~J. Pace, F.~Schapachnik, G.~Schneider, Contract automata - an
  operational view of contracts between interactive parties, Artif. Intell. Law
  24~(3) (2016) 203--243.

\bibitem{BDF16}
D.~Basile, P.~Degano, G.~L. Ferrari, Automata for specifying and orchestrating
  service contracts, Log. Methods Comput. Sci. 12~(4) (2016) 1--51.

\bibitem{ML18}
A.~Mavridou, A.~Laszka, Designing secure {Ethereum} smart contracts: {A} finite
  state machine based approach, in: S.~Meiklejohn, K.~Sako (Eds.), {FC} 2018,
  Revised Selected Papers, Vol. 10957 of LNCS, Springer, 2018, pp. 523--540.

\bibitem{MLSD19}
A.~Mavridou, A.~Laszka, E.~Stachtiari, A.~Dubey, Verisolid: {Correct-by-Design}
  smart contracts for {Ethereum}, in: I.~Goldberg, T.~Moore (Eds.), {FC} 2019,
  Revised Selected Papers, Vol. 11598 of LNCS, Springer, 2019, pp. 446--465.

\bibitem{BFHP17}
P.~Bottoni, A.~Fish, A.~Heu{\ss}ner, F.~Parisi{-}Presicce, Resource-aware
  policies, J. Vis. Lang. Comput. 38 (2017) 84--96.

\end{thebibliography}
%

\end{document}